\documentclass[journal]{IEEEtran}
\IEEEoverridecommandlockouts

\usepackage{widetext}

\usepackage[utf8]{inputenc} 
\usepackage[T1]{fontenc}    
\usepackage{hyperref}       
\usepackage{url}            
\usepackage{booktabs}       
\usepackage{amsfonts}       
\usepackage{nicefrac}       
\usepackage{microtype}      
\usepackage{lipsum}
\usepackage{graphicx}
\graphicspath{ {./images/} }

\usepackage{multicol} 
\usepackage{makecell}
\usepackage{multirow}
\usepackage{widetext}
\usepackage{url}
\usepackage{amsfonts}
\usepackage{algorithmic}
\usepackage{graphicx}
\usepackage{textcomp}
\usepackage{amsmath}
\usepackage{amsthm}
\usepackage{amssymb}
\usepackage{caption}
\usepackage{color}
\usepackage{xcolor}
\usepackage{float}

\usepackage{amsmath,amsfonts}
\usepackage{algorithmic}
\usepackage{array}
\usepackage[caption=false,font=normalsize,labelfont=sf,textfont=sf]{subfig}
\usepackage{textcomp}
\usepackage{stfloats}
\usepackage{url}
\usepackage{verbatim}
\usepackage{graphicx}
\def\BibTeX{{\rm B\kern-.05em{\sc i\kern-.025em b}\kern-.08em
    T\kern-.1667em\lower.7ex\hbox{E}\kern-.125emX}}
\usepackage{balance}

\usepackage{multicol}
\usepackage{widetext}
\usepackage{cite}
\usepackage{url}
\usepackage{amsfonts}
\usepackage{algorithmic}
\usepackage{graphicx}
\usepackage{textcomp}
\usepackage{amsmath}
\usepackage{amsthm}
\usepackage{amssymb}
\usepackage{caption}
\newtheorem{theorem}{Theorem}
\newtheorem{remark}{Remark}
\newtheorem{assumption}{Assumption}
\newtheorem{definition}{Definition}
\newtheorem{corollary}{Corollary}
\newtheorem{lemma}{Lemma}

\usepackage{color}
\usepackage{xcolor}
\usepackage{float}
\allowdisplaybreaks[1]
\usepackage{amsthm}
\usepackage{graphicx}

\usepackage{amsmath,amssymb,amsthm,thmtools}

\AtBeginDocument{%
  \providecommand\BibTeX{{%
    \normalfont B\kern-0.5em{\scshape i\kern-0.25em b}\kern-0.8em\TeX}}}

\begin{document}

\author{Yuetai Li, Zhangchen Xu, Yiqi Wang, Zihan Zhou, Lei Zhang,~\IEEEmembership{Senior~Member,~IEEE} and Jon Crowcroft,~\IEEEmembership{Fellow,~IEEE}

\thanks{Yuetai Li, Zhangchen Xu, are with Department of Electrical \& Computer Engineering,
University of Washington, 1410 NE Campus Pkwy, Seattle, WA, United States (e-mail: yuetaili@uw.edu, zxu9@uw.edu).}

\thanks{Yiqi Wang and 
Lei Zhang (corresponding author) are with James Watt School of Engineering, University of Glasgow, Glasgow, G12 8QQ, United Kingdom (e-mail: 2720884W@student.gla.ac.uk, Lei.Zhang@glasgow.ac.uk).}

\thanks{Zihan Zhou is with Information Hub of HKUST, Hong Kong University of Science and Technology, No.1 Du Xue Rd, Nansha District, Guangzhou (e-mail: zzh45472@gmail.com).}

\thanks{Jon Crowcroft is with Computer Lab, University of Cambridge, Cambridge, CB2 1TN, United Kingdom (e-mail: jon.crowcroft@cl.cam.ac.uk).}

}

\title{Distributed Consensus Network: A Modularized Communication Framework and Reliability Probabilistic Analysis}

\maketitle

\begin{abstract}
In this paper, we propose a modularized framework for communication processes applicable to crash and Byzantine fault-tolerant consensus protocols. We abstract basic communication components and show that the communication process of the classic consensus protocols such as RAFT, single-decree Paxos, PBFT, and Hotstuff, can be represented by the combination of communication components. 
Based on the proposed framework, we develop an approach to analyze the consensus reliability of different protocols, where link loss and node failure are measured as a probability. 
We propose two latency optimization methods and implement a RAFT system to verify our theoretical analysis and the effectiveness of the proposed latency optimization methods. We also discuss decreasing consensus failure rate by adjusting protocol designs. This paper provides theoretical guidance for the design of future consensus systems with a low consensus failure rate and latency under the possible communication loss.
\end{abstract}


\section{Introduction}

Distributed consensus, a fundamental concept of distributed systems, ensures that independent nodes agree on the same value to perform a certain task through message exchange within the network, in spite of the presence of faulty nodes. 
Characterized by distinct failure behaviors,  two failure models exist: crash failure, where a faulty node abruptly stops working without resuming \cite{184040, lamport2001paxos}, and Byzantine failure, where a faulty node acts arbitrarily in order to maximally damage the consensus \cite{castro1999practical,yin2019hotstuff}. Protocols that can resist crash failures and Byzantine failures are called crash fault-tolerant (CFT) and Byzantine fault-tolerant (BFT) protocols.

The consensus has been extensively researched and applied in traditional fields such as distributed databases for ensuring safety and liveness. In the research of these applications, most of consensus protocols were initially proposed in reliable communication networks, where there is a known or unknown upper bound on transmission delays and all messages are guaranteed to be delivered within this bound.

Recently, emerging distributed architectures such as blockchain networks have led to a resurgence in distributed consensus. 
\textcolor{black}{Some recent work has studied the applications of wireless consensus systems including blockchain-enabled IoT and autonomous systems \cite{9003220, 10041971, 10000789, 9119383, 8977452, luo2023symbiotic, 10414128}. \cite{9003220, 10041971,  luo2023symbiotic} focus on consensus in critical mission scenarios, e.g., autonomous driving systems wherein vehicles and pedestrians go through intersections, IoT environments where terminals (drones, sensors, and actuators) act based on coordinated behaviors. \cite{10000789} implements wireless RAFT on embedded system to make sensing data and action orders consistent in the Internet of Vehicles. \cite{9119383, 8977452, 10414128} consider blockchain-enabled IoT where IoT devices in the wireless mobile-edge network are supported by the set of blockchain peers that process IoT transactions.} 

\textbf{Motivation.}
In these applications such as blockchain-enabled IoT and Internet of Vehicles,  wireless communication is essential to support consensus operations in large-scale systems and facilitate connections among IoT devices and other clients \cite{SeoHyowoon2021CaCC, XuHao2020RBWB}. However, wireless transmission links might be unreliable and link loss may occur due to channel fading or spectrum jamming \cite{VojcicB.R1989Pods, TehK.C1998Mjro}. 
Even worse, wireless scenarios may provide an additional attack surface for Byzantine adversaries. In addition to classical Byzantine behaviors like deception or keeping silent, Byzantine nodes in wireless networks can further interfere with the communication of other participants (e.g., through broadband noise jamming or selective jamming), resulting in a compromised quality of communication. 

Such unreliable communications introduce new challenges. \textcolor{black}{As the consensus protocol depends on message exchanges to enhance interconnectivity among nodes and attain consistency, compromised communication can significantly impair the reliability of the distributed system.} 
The failure in communication links may lead to a failure in achieving consensus instances, ultimately compromising the overall performance of the system, e.g., a significant increase in latency. In Figure \ref{fig: intro_fig}, we find that a small reduction of quality in communication (link loss rate from 0 to 0.05) could greatly increase the latency of the consensus system. 

\begin{figure}[htbp]
  \centering 
  \includegraphics[width=7cm, height = 4.5cm]{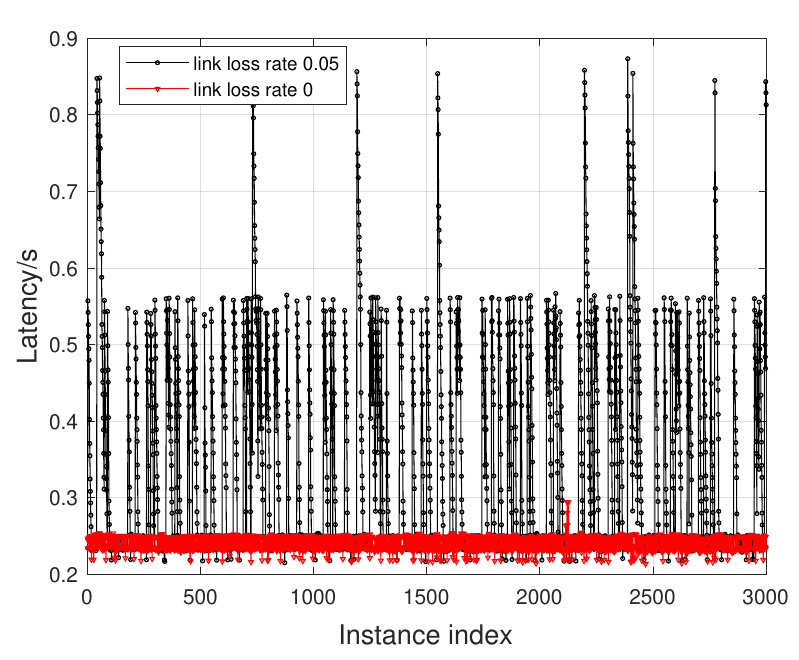}
  \caption{The comparison of latency in the RAFT between the link loss rate of 0 and 0.05. The system consists of 6 followers and 1 leader with 2 followers crashed. The latency is calculated from the instance's arrival at the leader to its commitment by the leader. The single link latency is set as 0.1s. The instance arrival rate is set as 60 per second.}
  \label{fig: intro_fig} 
\end{figure}

In light of finding this, our paper's motivation lies in the crucial need for modeling the performance of distributed consensus protocols under compromised communication modules such as link loss. \textcolor{black}{Regarding many applications of using consensus protocols such as RAFT and PBFT in wireless setting \cite{9003220, 10041971, 10000789, luo2023symbiotic}, this is critical to evaluate how reliable the existing protocols are when they are used in the wireless network \cite{SeoHyowoon2021CaCC, XuHao2020RBWB}.} However, since consensus protocols have different communication processes and protocol specifications such as PBFT \cite{castro1999practical} and Hotstuff \cite{yin2019hotstuff}, the analytical method of one protocol may not be able to extend to other protocols \cite{SeoHyowoon2018CDDS, 9003220}.

Therefore, this paper aims to propose a modularized framework for representing communication structures across a wide range of consensus protocols.

\textbf{Our Contribution.}
In this paper, we propose a modular approach based on a set of fundamental communication processes (which we refer to as communication components henceforth) and utilize them to represent the communication of various protocols.
Subsequently, the protocol analysis can attain a higher level of scalability, thereby facilitating the evaluation of the communication process for extant protocols.
Furthermore, we evaluate probabilistic consensus-achieving rates, enabling comparisons between existing protocols and giving theoretical guidance into the designs of future protocols and systems. The main contributions of this paper are summarized as follows:

\begin{enumerate}
    \item  We propose a modularized framework to represent the communication process of various consensus algorithms. Specifically, we present fundamental communication components in consensus algorithms and employ their combinations to represent some well-established protocols, such as RAFT \cite{184040}, single-decree Paxos \cite{lamport2001paxos}, PBFT \cite{castro1999practical}, and Hotstuff \cite{yin2019hotstuff}.
    
    \item Utilizing the proposed modularized framework, we examine consensus reliability in the context of potential node and link failures through a probabilistic lens. \textcolor{black}{We further propose the concepts of Reliability Gain and Tolerance Gain for different protocols, which shows that the consensus failure rate has linear relationships with two fundamental system parameters, i.e., the overall joint failure rate and the maximum number of faulty nodes the system can tolerate.}

    \item We show that a larger consensus failure rate leads to higher transmission and queuing latency. We subsequently proposed two methods based on our theoretical analysis of consensus failure rate to decrease the consensus failure rate, and thus optimize the system latency. We also discuss decreasing consensus failure rate by adjusting protocol designs.

    \item \textcolor{black}{We implement a RAFT consensus system as an example to verify our theoretical analysis and show the effectiveness of the proposed latency optimization methods. Additionally, we demonstrate the numerical results of the consensus failure rate for different consensus protocols.}
    
\end{enumerate}

The structure of this paper is as follows. The related work is shown in Section \ref{sec: related work}. Section \ref{Communication Structures of Consensus Protocols} proposes a generic and modularized framework to represent the communication structure of different consensus protocols. Based on this framework, the probability analysis of consensus reliability for different protocols and the concepts of Reliability Gain and Tolerance Gain are proposed in Section \ref{rv_derivation}. 
Section \ref{sec: consensus reliability and latency} shows a larger consensus failure rate causing higher latency and proposed latency optimization methods.  Section \ref{sec: experiments} implements a Raft system as an example to verify the theoretical analysis.
Section \ref{Discussion: Improve Consensus Reliability by Protocol Design} discusses improving consensus failure rate by adjusting protocol designs. Finally, Section \ref{sec: conclusion and future work} makes a conclusion.

\section{Related Work}
\label{sec: related work}

\subsection{Distributed Consensus:}
\textcolor{black}{Most consensus protocols were initially proposed in reliable communication networks \cite{10.1145/357172.357176,doi:10.1137/0212045,feldman1997optimal,10.1145/42282.42283,184040,lamport2001paxos,castro1999practical,yin2019hotstuff}. They follow classical synchronous or partial synchronous communication model, where there is a known or unknown upper bound on message transmission delays between nodes and all messages are guaranteed to be delivered within this time bound.} As for unreliable link communications such as link loss, some work \cite{doi:10.1137/S009753970443999X, 6312888,1238089} prove the properties of agreement, termination, and validity in designing protocols. \cite{10.1145/3447851.3458739} also models the partial network partition of RAFT and analyzes possible outage caused by possible link loss. However, the analysis often treats links in a deterministic manner, i.e., either faulty or non-faulty. There is few probabilistic analysis of link and consensus reliability. 
In \cite{doi:10.1137/S009753970443999X}, assumption coverage has been proposed as a measure of whether the probability of meeting the required conditions for protocols converges to 1, which is a meaningful measure for assessing the quality of protocols. Nonetheless, this work provides only a vague evaluation of consensus reliability boundaries as the number of nodes approaches infinity and assumes identical link failure rates for different links. 

\subsection{Availability of Quorum System: } Some work analyze the availability of quorum system \cite{Peleg1995TheAO, doi:10.1137/S0097539797325235}. This characterizes the likelihood that a quorum system will be able to deliver service, taking into account the failure probability of various elements within the system. While many consensus algorithms can be represented as quorum-based systems, quorum system availability may not be used to assess the consensus networks reliability, as it usually does not account for communication link failures within the network.

\subsection{Wireless Consensus Applications: }Some recent works have proposed many applications of wireless distributed consensus and analyzed the system performance \cite{9003220, AsheralievaAlia2020RCFf, SeoHyowoon2018CDDS, KimHyesung2020BOFL, LiuYinqiu2019mALB, WangWenbo2019ASoC, ZhangLei2021HMCR, 10041971}. 
\cite{9003220, 10041971, luo2023symbiotic} focus on consensus in critical mission scenarios, e.g., autonomous driving systems wherein vehicles and pedestrians go through intersections, IoT where terminals (drones, sensors, and actuators) act based on coordinated behaviors. \cite{10000789} implements RAFT based on embedded system to make sensing data and action orders consistent in the Internet of Vehicles. 
\cite{XuHao2020RBWB, SeoHyowoon2021CaCC} 
has analyzed the impact of wireless link reliability on consensus reliability and latency, but their models target the geographic distribution of nodes for the specific protocol.
\cite{9184829,9990047, Li} has investigated RAFT consensus reliability in a wireless network to show critical decision-making with different link reliability in IoT. However, the proposed result was only for RAFT.

\section{Communication Structures of Consensus Protocols}
\label{Communication Structures of Consensus Protocols}
In this section, we propose a modularized framework to model the communication process of different consensus protocols. We first give the models and assumptions, and propose basic consensus communication components. Then the communication structure of typical consensus algorithms including RAFT, PBFT and Hotstuff could be represented by combinations of basic communication components.

\subsection{Models and Assumptions}
\label{sec: models and assumptions}
The consensus network consists of nodes communicating with each other to achieve consistency based on the designed protocols.
A faulty node exhibits unexpected or incorrect behavior, which can be due to hardware malfunctions, software bugs, attacks, or network issues. A non-faulty node, on the other hand, is a node that functions correctly and as expected. The primary-backup paradigm is a common terminology in distributed consensus protocols, where one node is the primary who coordinates the consensus process, while the other nodes, known as backups, update states to maintain overall system consistency and fault tolerance. In our paper, we denote number of backups as $n$ and the maximum number of faulty nodes the consensus protocol can tolerate as $f$. In most CFT protocols, $f=\lfloor n/2 \rfloor$, and in most BFT protocols, $f=\lfloor n/3 \rfloor$.

The consensus protocols typically have a normal path (i.e., \textit{normal case operation}), following a propose-vote paradigm. When the primary misbehaves, there is another recovery path (often called \textit{view change}) to rotate the primary. 
Following \cite{Li, XuHao2020RBWB, 9003220}, we analyze the network reliability of normal case operations. 
If the primary node fails, the protocol will automatically perform a view change \cite{184040,lamport2001paxos,castro1999practical,yin2019hotstuff}, find a new primary node, and continue with normal case operations.

Regarding node failure, we adopt the classical static corruption model \cite{shi2020foundations} which assumes that the set of faulty nodes is chosen at the start of the protocol but the protocol is unaware of which nodes are corrupt. We assume that all backups are not 100\% reliable (e.g., limited by cost, size, and complexity), thus have a probability to fail in a consensus network. \textcolor{black}{In addition,} faulty nodes are assumed to be not recovered after the failure.

For the communication links, we assume each link in the consensus process will loss as probability in a wireless connected network \cite{SeoHyowoon2021CaCC, Li, 9003220}. This probability setting aligns more closely with real-world wireless networks, since in classic wireless research the link reliability is commonly modeled as a probability issue \cite{1374908}. Regarding the necessary wireless setting in many consensus applications such as blockchain-enabled IoT \cite{9003220, 10041971, 10000789, luo2023symbiotic}, probability modeling of link loss is critical to evaluate the reliability of existing protocols \cite{SeoHyowoon2021CaCC, Li}.
Note that this is also necessary in the BFT setting. Although the BFT assumption includes the scenario where Byzantine nodes may deliberately lose messages, it is still necessary to model the link loss among the remaining honest nodes.

\subsection{Node Activation}
\label{sec: Node Activation}
A consensus protocol consists of a few phases (aka, step or stage represented in different protocols) of message exchange, where different protocols have various communication phases. Within each phase, a non-faulty node usually has a local verification condition to check the validity of the received message and decide whether to proceed to the next phase of normal case operations. For example, in the prepare phase of PBFT, a non-faulty node proceeds only if it collects more than 2/3 node's correct and valid messages from other nodes.

We define the concept of \textbf{Activated Node (AN)} as follows, which characterizes the procedure of local verification for entering the next phase.
\begin{definition}[Activated Node]
A node will become an activated node if it has met the protocol condition to proceed and follows the normal case operation. 
\end{definition}

We will also call a node \textbf{Inactivated Node (IN)} or not becoming an activated node if it fails to meet the protocol conditions (e.g., not receiving enough required messages from other nodes or primary), or fails to follow the normal case operation in this phase. \textcolor{black}{Particularly, according to the definition of AN and the static corruption model, if we denote the start of the protocol as phase 0, non-faulty nodes could be regarded as ANs in phase 0 and faulty nodes could be regarded as IN in phase 0.}

\subsection{Basic Communication Components}
\label{sec: Basic Communication Components}
\textcolor{black}{In this section, we use a set of fundamental communication processes (which we refer to as communication components henceforth) to represent the communication of message exchange in each phase of the consensus protocols.} 

\begin{figure}[htbp]
  \centering 
  \includegraphics[width=0.3\textwidth, height = 5cm]{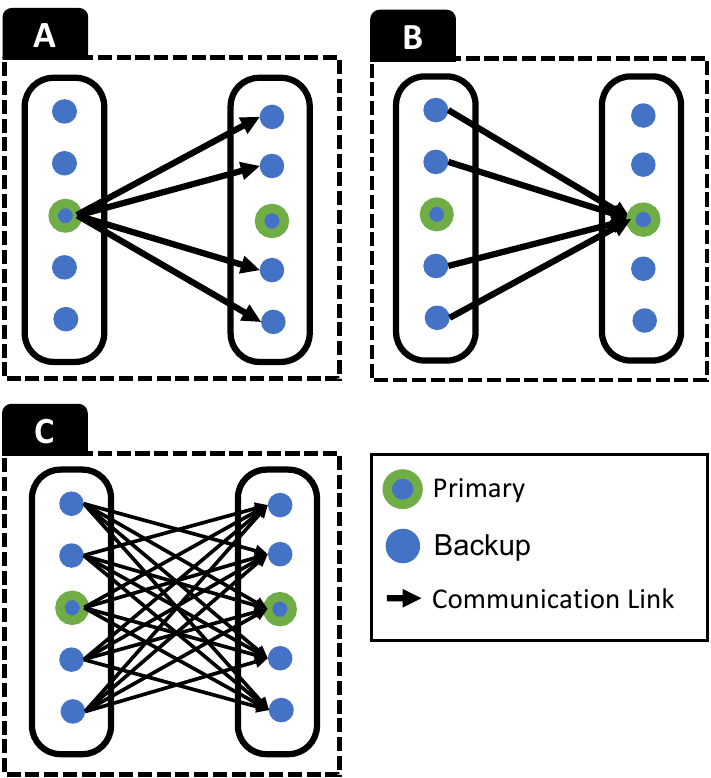}
  \caption{\textcolor{black}{Three Representative Communication Designs in Consensus Protocols}}
  \label{fig: communication graph} 
\end{figure}

We show three representative communication graphs in Figure \ref{fig: communication graph}, i.e., one-to-many, many-to-one, and many-to-many.
We will call them A, B, and C respectively in the remainder of our paper. 
Let ${(T_j)}_{r_j}^{M_j}, T_j\in \{A,B,C\}$ denote a \textbf{communication component} that represents the communications of message exchange at the $j$-th phase of the consensus protocols. The notation is explained as follows:
\begin{enumerate}
    \item $M_j$ represents that the protocol requires at least $M$ activated nodes at the end of the $j$-th phase. Usually, $M_j$ takes either $n - f$ or $f + 1$.
    \item $r_j$ represents the phase dependence relationship. Only those nodes that have been activated in the $(j-r_j)$-th phase 
    could be activated in the $j$-th phase. This implies that the $j$-th phase is dependent on the $(j-r_j)$-th phase.
    \item $T_j$ represents the communication graph in the $j$-th phase, i.e., one-to-many for $T_j=A$, many-to-one for $T_j=B$, and many-to-many for $T_j=C$. Different $T_j$ also has various conditions of node activation as:
    
\begin{itemize}
    \item $T_j=A$: the node should successfully receive and verify the message from the primary.
    \item $T_j=B$: the node should successfully send the message to the primary. 
    \item $T_j=C$: the node should successfully receive and verify no less than $n-f-1$ messages from other nodes that have been activated in the last phase.
\end{itemize}
    
\end{enumerate}

\textcolor{black}{These three parameters indicate the relationship between a communication component and node activation. $M_j$ shows the minimum number of ANs required in the phase, while $r_j$ and $T_j$ show the conditions of a node to become AN in the phase. We call that $r$ represents \textit{Node Activation Condition I} and $T$ represents \textit{Node Activation Condition II}. The following will give more explanation of these three parameters.} 

\begin{table*}[ht]
\centering
\renewcommand\arraystretch{1.4}
\caption{Summary of communication structures of typical consensus algorithms.}
\begin{tabular}{|c|c|c|}
\hline
\textbf{Protocol} & \textbf{Category} & \textbf{Communication Structure $G$} \\
\hline
Single-decree Paxos & CFT & $G=[A_{1}^{n-f},B_{1}^{n-f},A_{3}^{n-f},B_{1}^{n-f}]$  \\
\hline
Raft & CFT & $G=[A_1^{n-f},B_1^{n-f}]$ \\
\hline
\textcolor{black}{PBFT} & BFT & $G=[A_{1}^{n-f},C_{1}^{n-f},C_{1}^{f+1},B_{1}^{f+1}]$ \\
\hline
Hotstuff & BFT & $G=[A_1^{n-f},B_1^{n-f},A_2^{n-f},B_1^{n-f},A_2^{n-f},B_1^{n-f},A_2^{n-f},B_1^{f+1}]$ \\
\hline
\end{tabular}
\label{table: popular consensus algorithm summary}
\end{table*}

We use the example of PBFT to explain $M_j$. \textcolor{black}{In the \textit{prepare} phase of PBFT, at least $n-f-1$ messages from others are needed to be received by one node. This requires there should be at least $n-f$ activated nodes at the end of the \textit{pre-prepare} phase (the phase before \textit{prepare} phase). Thus $M_1=n-f$ in \textit{pre-prepare} phase.} Differently, only $f+1$ activated nodes are needed at the end of the \textit{reply} phase to respond to the client. Therefore, for the \textit{reply} phase $M_4=f+1$. Please refer to Appendix \ref{sec: structure of PBFT} for more PBFT protocol details.

$r_j$ implies that the current phase is dependent on the $(j-r_j)$-th phase. Although naturally a node would be activated if it satisfies the Node Activation Condition II based on $T_j$, we emphasize that $r_j$ and the Node Activation Condition I are necessary, i.e., ANs in the current phase are required to have been ANs in the $(j-r_j)$-th phase. The reason behind this is that some protocols require strict verification of incoming messages that rely on previous phases. If the node suffers link loss in the previous phases (i.e., the node has not been activated previously), there will be no adequate information for it to meet the protocol condition to proceed. 

Take PBFT, Hotstuff, and single-decree Paxos as examples. In PBFT, the \textit{prepare} phase requires nodes to have received \textit{pre-prepare} messages from the primary to verify the \textit{prepare} messages. Since the \textit{prepare} phase relies on the \textit{pre-prepare} phase that is 1 phase before, $r_2$ is taken as 1. 
In Hotstuff, if a backup has not received the \textit{pre-commit} message from the primary yet but received the \textit{commit} message from the primary, the type of the received message cannot match the required message type so that the node cannot be activated. Here $r_3$ is taken as 2 since there is a many-to-one component and a one-to-many component between \textit{pre-commit} messages and \textit{commit} messages.
One more example is the \textit{Propose} phase of single-decree Paxos. As long as nodes are non-faulty, they can proceed when receiving the \textit{Propose} message from the proposer even though they missed the \textit{Prepare} message before. \textcolor{black}{Recall that we regard non-faulty as ANs in phase 0 (the start of the protocol) and faulty nodes as INs in phase 0. Since the \textit{Propose} phase is the third phase of single-decree Paxos, $r_3$ could be taken as 3 in this phase.}  Please see Appendix \ref{apd: Communication Structure of Representative Consensus Algorithms} for more protocol details and explanations.

\subsection{Communication Structure of Representative Consensus Algorithms}
\label{Communication Structure of Representative Consensus Algorithms}

To our knowledge, using the combinations of communication components introduced in Section \ref{sec: Basic Communication Components} can well represent the communication process of many consensus algorithms. As shown in Table \ref{table: popular consensus algorithm summary}, several examples of typical consensus algorithms are given including RAFT \cite{184040}, Paxos \cite{lamport2001paxos}, PBFT \cite{castro1999practical}, and Hotstuff  \cite{yin2019hotstuff}. The detailed analysis and explanation\footnote{We consider the communication structure of proposed protocols in the original papers as our analyzing objects, while the communication structure of modified versions might differ from our results. However, the analyzing way is similar.} are shown in Appendix \ref{apd: Communication Structure of Representative Consensus Algorithms}.

\section{Consensus Reliability Analysis}
\label{rv_derivation}
For real-world applications, especially in the engineering field, it is unrealistic to analyze the link in a deterministic manner, i.e., either faulty or non-faulty, since all links might not be secured. From the probability view, this section analyzes a fundamental concept for consensus protocols: \textbf{consensus reliability} \cite{Li, 9990047}. This is defined as the probability of consensus success for a consensus protocol when each node has a probability to become faulty in the system initialization, and each link has a probability to become loss. Based on the communication components introduced in Section \ref{Communication Structures of Consensus Protocols}, we propose a theoretical framework to analyze consensus reliability for different protocols. The calculated consensus reliability could be considered as the lower bound of the rate of consensus achieving, i.e., Byzantine nodes are assumed to cause the worst-case scenarios. \textcolor{black}{We should mention that our proposed consensus reliability does not aim at replacing a direct design and proof framework of a distributed consensus algorithm. However, it is fundamental for reliability analysis and engineering evaluation in the consensus network with unreliable nodes and links.}

\subsection{Notations}
\label{Representation by node reliability or link reliability}
\textcolor{black}{Let $G_j={(T_j)}_{r_j}^{M_j}$} denote the $j$-th phase of a protocol $G$, where $T_j,r_j,M_j$ follows the definitions in Section \ref{sec: Basic Communication Components} and $j$ satisfies $j\le|G|$. Here $|G|$ represents the number of communication components of protocol $G$. 

Let $\Omega=\{Node_1,Node_2,\ldots,Node_n\}$ represent the set of $n$ backups in the cluster. When analyzing the consensus process under possible node failure and link failure, it is important to consider the set of activated nodes in each phase. We denote $ S^{G_j}_{x_j}\subseteq\mathrm{\Omega}$, with the cardinality $\left|S_{x_j}^{G_j}\right|=x_j$, as the set of activated nodes
at the end of the $j$-th phase of the protocol $G$.
Let the probability of node $i$ becoming AN at the end of the $j$-th phase be $P_i^{G_j}$, where $i$ represents the $Node_i \in \Omega$. We could also call $P_i^{G_j}$ as the \textbf{Activation Probability} of node $i$ in the $j$-th phase.

The following shows \textbf{how the Activation Probability $P_i^{G_j}$ could be expressed by Node/Link failure rate}: 

When $j=0$, recall that $S^{G_0}_{x_0}$ represents the set of non-faulty nodes in the system. Thus $P_i^{G_0}$ is the probability that node $i$ is non-faulty. For $j=1, 2...,|G|$, according to the Node Activation Conditions I and II in Section \ref{sec: Basic Communication Components}, the following will obtain the expression of $P_i^{G_j}$ when $T_j$ is the graph A, B, and C respectively.

\textbf{One-to-Many (A):}  The nodes in the $j$- th phase could be activated if they are in the set $S_{x_{j-r_j}}^{G_{j-r_j}}$ and they successfully receive the message from the primary. Thus we can express $P_i^{G_j}$ as
$P_i^{G_j}=P^{primary,i}_{L}$, where $P^{primary,i}_{L}$ is 
the reliability of communication link from the primary to the node $i$ which are in the set $S_{x_{j-r_j}}^{G_{j-r_j}}$. 

\textbf{Many-to-One (B):} According to the node activation condition of B mentioned in Section \ref{sec: Basic Communication Components}, the probability of becoming AN can be obtained as $P_i^{G_j}=P^{i,primary}_{L}$, where $P^{i,primary}_{L}$ represents the reliability of communication link from the node $i$ in the set $S_{x_{j-r_j}}^{G_{j-r_j}}$ to the primary.

\textbf{Many-to-many (C):} 
If the graph is C, $P_i^{G_j}$ is the probability that node $i$ in $S_{x_{j-r_j}}^{G_{j-r_j}}$ receives at least $n-f-1$ messages from other activated nodes in $S_{x_{j-r_j}}^{G_{j-r_j}}$.
Thus $P_i^{G_j}$ can be calculated as:
\begin{equation}
\label{rv_1A_eq_5}
\begin{split}
P_i^{G_j}=&\sum_{\substack{x_{j-r_j}-1\geq k\geq n-f-1\\S_k\subseteq (S_{x_{j-r_j}}^{G_{j-r_j}}\setminus i) }}{\prod_{u\in S_k} P^{u,i}_{L}}\prod_{v\in (S_{x_{j-r_j}}^{G_{j-r_j}}\setminus \{i,S_k\})}\left(1-P^{v,i}_{L}\right)
\end{split}
\end{equation}
where $P^{u,i}_{L}$ represents the reliability of communication link from the node $u$ to $i$, $\setminus$ is the setminus, and $k,u,v$ are running variables. $S_k$ is a running variable with $k=\left|S_k\right|$

\subsection{Main Theorem of Consensus Reliability}
\label{Markov Property of Consensus}
In this section, we obtain the theoretical consensus reliability in Theorem \ref{th1} according to the communication structures of different consensus protocols.

\textcolor{black}{We will consider the probability Pr(the set of activated nodes in j-th phase is $S^{G_j}_{x_j}$) and simplify the expression as $P(S^{G_j}_{x_j})$. Similarly, the conditional probability $P(S^{G_j}_{x_j}|S^{G_i}_{x_i})$ represents given the set of activated nodes in i-th phase as $S^{G_i}_{x_i}$, the probability that the set of activated nodes in j-th phase is $S^{G_j}_{x_j}$.}
\textcolor{black}{ According to the Node Activation Condition I in Section \ref{sec: Basic Communication Components}, 
it could be directly obtained that the set of activated nodes in $j$-th phase is only dependent with that in $(j-r_j)$-th phase. This implies Markov property and could be formally expressed as the following:}
\begin{lemma} 
 ANs in $j$-th phase is the subset of ANs in $(j-r_j)$-th phase, i.e.,
\label{rv_p_1}
\begin{equation}
S_{x_j}^{G_j} \subseteq S_{x_{j-r_j}}^{G_{j-r_j}}\ \ (x_{j}\le x_{j-r_j})
\end{equation}
\end{lemma}
\begin{lemma}
\label{rv_p_2}
Given $S_{x_0}^{G_0},S_{x_1}^{G_1},\ldots,S_{x_{j -1}}^{G_{j-1}}$, ANs in the $j$-th phase only depends on ANs in the $(j-r_j)$-th phase, i.e.,
\begin{equation}
P\left(S_{x_j}^{G_j}\middle| S_{x_0}^{G_0},S_{x_1}^{G_1},\ldots,S_{x_{j -1}}^{G_{j-1}}\right)=P\left(S_{x_j}^{G_j}\middle| S_{x_{j-r_j}}^{G_{j-r_j}} \right)
\end{equation}
\end{lemma}

If $r_j=1$ for any $j=1,...,|G|$, the Markov property will be first-order. Otherwise it will be high-order. Now we could obtain our main theorem of the theoretical consensus reliability as follow:
\begin{theorem}
\label{th1}
Given the communication structure $G$ of the protocol, the consensus reliability, denoted as $P_C$, can be calculated as:
\begin{equation}
\label{eq_th1}
\begin{split}
P_C=&\sum_{\substack{n\geq x_0 \geq n-f\\S_{ x_0}^{G_0}\subseteq \mathrm{\Omega}}}\!...\! \sum_{\substack{x_{j-r_j}\geq x_j \geq M_j\\S_{ x_j}^{G_j}\subseteq S_{x_{j-r_j}}^{G_{j-r_j}}}} 
\! ... \! \sum_{\substack{x_{|G|-r_{|G|}}\geq x_{|G|} \geq M_{|G|}\\S_{ x_{|G|}}^{G_{|G|}}\subseteq S_{x_{|G|-r_{|G|}}}^{G_{|G|-r_{|G|}}}}}
\\ &
P\left(S_{x_0}^{G_0},S_{x_1}^{G_1},...,S_{x_{\left|G\right|}}^{G_{\left|G\right|}}\right)
\end{split}
\end{equation}
where 
\begin{equation}
\label{rv_eq_0}
P\left(S_{x_0}^{G_0},S_{x_1}^{G_1},\ldots,S_{x_{\left|G\right|}}^{G_{\left|G\right|}}\right)=P\left(S_{x_0}^{G_0}\right)\prod_{j=1}^{\left|G\right|}P(S_{x_j}^{G_j}|S_{x_{j-r_j}}^{G_{j-r_j}})
\end{equation}
\begin{equation}
\label{rv_1A_eq_3}
P\left(S_{x_0}^{G_0}\right)=\prod_{u\in S_{x_0}^{G_0}} P_u^{G_0}\prod_{v\in\complement_\mathrm{\Omega} S_{x_0}^{G_0}}\left(1-P_v^{G_0}\right)
\end{equation}
\begin{equation}
\label{rv_1A_eq_4}
P\left(S_{x_j}^{G_j}\middle| S_{x_{j-r_j}}^{G_{j-r_j}}\right)=\prod_{u\in S_{x_j}^{G_j }} P_u^{G_j}\prod_{v\in\complement_{S_{x_{j-r_j}}^{G_{j-r_j}}}S_{x_j}^{G_j}}\left(1-P_v ^{G_j}\right)
\end{equation}
here $P_i^{G_j}$ is the activation probability of node $i$ in the $j$-th phase defined in Section \ref{Representation by node reliability or link reliability}, and $\complement$ is the notation of complement set.
\end{theorem}
Please see Appendix \ref{Proof: Theorem th1} for proof.

As mentioned in Section \ref{Representation by node reliability or link reliability},  
$P_i^{G_j}$ can be represented by node/link reliability. Thus in Theorem \ref{th1}, $P_C$ is a function of node/link reliability (or node/link failure rate) at different phases. Based on Theorem \ref{th1}, for consensus protocols listed in Table \ref{table: popular consensus algorithm summary} or other protocols of which the communication structure can be presented in Section \ref{Communication Structure of Representative Consensus Algorithms}, we can theoretically analyze the influence of the reliability of each node/link on the final consensus reliability. 

\begin{remark}
\textcolor{black}{Assuming different consensus instances are independent, we show that the consensus reliability of consecutive multiple instances can be extended based on that of a single instance shown in Theorem \ref{th1}. 
It can be used to evaluated how many instances will cause at least one consensus instance failure with probability 1. 
Please see Appendix \ref{consecutive instances} for more analysis.}
\end{remark}

\begin{remark}
\label{dependent_case}
\textcolor{black}{If $P_i^{G_j}$ is not independent for different node $i$ (corresponding the reliability of node/link of different nodes are not independent, e.g., correlation loss of wireless communication), Eq. (\ref{rv_1A_eq_3}) and (\ref{rv_1A_eq_4}) should be replaced as:}
\begin{equation}
\label{rv_1A_eq_1}
\begin{split}
P\left(S_{x_0}^{G_0}\right) = & \Pr ( ( \bigcap_{u \in S_{x_0}^{G_0}} u \text{ is non-faulty} ) \bigcap \\
& ( \bigcap_{v \in \complement_\Omega S_{x_0}^{G_0}} v \text{ is faulty} ) )
\end{split}
\end{equation}

\begin{equation}
\label{rv_1A_eq_2}
\begin{split}
P\left(S_{x_j}^{G_j}\middle| S_{x_{j-r_j}}^{G_{j-r_j}}\right)=& \Pr ((\bigcap_{u\in S_{x_j} ^{G_j}}{\text{ u is activated}})\bigcap\\
&{(\bigcap_{v\in\complement_{S_{x_{j-r_j}}^{G_{j-r_j}}}S_{x_j} ^{G_j}}{\text{v is not activated}}}))
\end{split}
\end{equation}
Then chain rule can be used to calculate Eq. (\ref{rv_1A_eq_1})-(\ref{rv_1A_eq_2}). Therefore, based on the conditional probabilities, the correlated probability $P_i^{G_j}$ could be fully characterized as long as there is sufficient prior information about the reliability of each node/link. 
\end{remark}

\subsection{Consensus Reliability Simplification}
\label{appro_sim}

Section \ref{Markov Property of Consensus} analyzes the theoretical consensus network reliability. However, the probabilistic computation is complicated, especially when the link/node reliability at different phases are varied. In this section, we simplify the calculations of consensus reliability, which is shown in Theorem \ref{multiAas1}.

\begin{theorem}
\label{multiAas1}
Given the communication structure $G$ 
, consider all the dependence relationships in $G$ as a graph where each phase is a node and each phase dependence relationship $r_j$ is an edge. All the dependence relationships form a tree, where the 0-th phase is the root node. Assume that there are $L$ leaf nodes and let $R_l(l=1,2,...,L)$ denote the set of phases in the path between the root node and the $l$-th leaf node, the consensus failure rate could be approximated as:
\begin{equation}
P_F\approx\sum_{l=1}^{L}P_F^{R_l}
\end{equation}
where $P_F^{R_l}$ is
\begin{equation}
P_F^{R_l}
= \sum_{\mathrm{\Omega}\supseteq S_{f+1}^{JF,R_l}}\prod_{u\in S_{f+1}^{R_l}} P_u^{ R_l}
\end{equation}
where $P_u^{JF,R_l}=1-\prod_{k \in R_l}P_u^{G_k}$ and $P_u^{G_k}$ is the probability that $u$ becoming AN in the $k$-th phase.  
The approximation will be tight if the node/ link failure rates are close to 1. 
\end{theorem}

Please see Appendix \ref{proof of multiAas1} for proof.

This theorem intuitively explains the principle of simplifying high-order Markov (for any $r_j$) consensus reliability: $P_F$ of high-order Markov consensus can be approximately regarded as the summation of the $P_F^{R_l}$, i.e., consensus failure rates of multiple first-order Markov (for all $r_j = 1$) communication structures. Then the simplification methods for first-order Markov property shown in Lemma \ref{rv_th_jfv} and \ref{rv_th_ps} in Appendix \ref{proof of multiAas1} can be used for each $P_F^{R_l}$.

The accuracy of the approximations in Theorem \ref{multiAas1} will be shown in the numerical results in Section \ref{Numerical Results of Consensus Failure Rate}.

\subsection{Reliability Gain and Tolerance Gain}
\label{sec: Reliability Gain and Tolerance Gain}
In this section, we propose the concepts of \textbf{Reliability Gain} and \textbf{Tolerance Gain} for different consensus protocols to obtain more concise and intuitive results to deploy wireless consensus systems in engineering. 

\subsubsection{\textcolor{black}{Assumption of Identical Failure Rates of Different Nodes}}
\textcolor{black}{We give one more assumption \cite{doi:10.1137/S009753970443999X,9184829}, which is made to simplify the analysis from the overall link reliability and overall node reliability to obtain more intuitive properties. } 

\begin{assumption}
\label{assption:iid}
\textcolor{black}{We assume the node reliability of different nodes are identical to $p_N$ and that the link reliability of different links are identical to $p_L$. }
\end{assumption}

Let $p_{G_j}$ be the activation probability of each node becoming AN in the $j$-th phase. When $j=0$, the condition for the node to become AN in phase 0 is that the node is non-faulty, thus $p_{G_0}=p_N$, where $p_N$ is the probability that one node is no-faulty. When $j=1, 2, ... |G|$, $p_{G_j}$ is determined by the communication component of this phase. According to the analysis of Section \ref{Representation by node reliability or link reliability}, if the communication graph of the $j$-th phase is A or B, $p_{G_j}=p_L$. 
Let the binomial distribution operator $B(a,b,p)=\binom{a}{b}p^b\left(1-p\right)^{a-b}$. If the communication graph of the $j$-th phase is C, 
$p_{G_j}$ is
\begin{equation}
\label{mv_D}
p_{G_j}=\sum_{k=n-f-1}^{x_{j-r_j}-1}{B(x_{j-r_j}-1, k, p_L)}
\end{equation}
and the approximated form of $p_{G_j}$ is 
\begin{equation}
\label{mv_D_approx}
\bar{p}_{G_j}\approx \sqrt[f+1]{\frac{\sum_{a=n-f}^{n}(\sum_{k=n-f-1}^{a-1}B(a-1, k, p_L))^{f+1}}{f+1}}
\end{equation}
Here Eq. (\ref{mv_D}) is the degenerated form of Eq. (\ref{rv_1A_eq_5}) \textcolor{black}{while Eq. (\ref{mv_D_approx}) is the degenerated form of Eq. (\ref{as3}) under the Assumption \ref{assption:iid}.}

Let $r_0=1$ and $x_{-1}=n$. Theorem \ref{th1} and \ref{multiAas1} can be degenerated under the Assumption \ref{assption:iid}:

\begin{corollary}
\label{cor_d_th1}
The degenerated form of Theorem \ref{th1} is:
\begin{equation}
\label{mv_cr}
p_C=\sum_{x_0= n-f}^{n}...\sum_{x_j= M_j}^{x_{j-r_j}}...\sum_{x_|G|= M_{|G|}}^{x_{|G|-r_{|G|}}}\prod_{j=0}^{|G|}{B(x_{j-r_j},x_j,p_{G_j})}
\end{equation}
\end{corollary}

\begin{corollary}
\label{degen_simp}
Given the communication structure $G$ where $M_j=n-f$ for any $j=1,2,...|G|$, the degenerated form of Theorem \ref{multiAas1} is:
\begin{equation}
p_F\approx \binom{n}{f+1}p_{JF,R}^{f+1}
\end{equation}
where $p_{JF,R}=\sqrt[f+1]{\sum_{l=1}^{L}p_{JF,R_l}^{f+1}}$ and $p_{JF,R_l}=1-\prod_{k \in R_l}p_{G_k}$. We define $p_{JF,R}$ as the overall joint failure rate.
The approximation will be tight if the node/ link failure rates are close to 1.
\end{corollary}

\subsubsection{Reliability Gain}

According to Corollary \ref{degen_simp}, it can be seen that the consensus failure rate and the overall joint failure rate are approximately linear in logarithmic form. \textcolor{black}{We define the \textbf{Reliability Gain} as the decrement of the logarithmic consensus failure rate if the overall joint failure rate is reduced by one order of magnitude. We show that the Reliability Gain is equal to $f+1$ in the following theorem:}
\begin{theorem}
	\label{the:RG}
	When the overll joint failure rate $p_{JF,R}$ is reasonably small\footnote{Our results show that $p_{JF, R}=0.1$ is small enough to make the conclusion, while this condition is achieved in most of the application scenarios.}, the consensus failure rate $p_F$ has a linear relation of ${\rm log}p_{JF,R}$ in logarithmic form, 
	\begin{equation}
		\label{equ:linear_p_uni}
		{\rm log}p_F=k_p\cdot {\rm log}p_{JF,R}+h_p 
	\end{equation}
	where the Reliability Gain $k_p=f+1$, the intercept $h_p={\rm log}(\binom{n}{f+1})$.
\end{theorem}
The overll joint failure rate $p_{JF,R}$ is expressed by the link/node relaibilty, and Table. \ref{table: equations summary} shows $p_{JF,R}$ for different protocols. Theorem \ref{the:RG} shows that for a fixed number of nodes in a consensus network, the order of magnitude of consensus failure rate ${\rm log}p_F$ decreases linearly with the coefficient $f+1$ by reducing ${\rm log}p_{JF,R}$, which can be obtained by increasing the communication quality or using facilities with higher reliability. Thus the proposed Reliability Gain can be an intuitive indicator for evaluating the consensus reliability.

\subsubsection{Tolerance Gain}

Furthermore, \textcolor{black}{we define the \textbf{Tolerance Gain} as the decrement of the logarithmic consensus failure rate if the maximum number of tolerant faulty nodes $f$ increases 1.
The following theorem gives a closed-form of the \textbf{Tolerance Gain}.}	

\begin{theorem}
	\label{the:RA}
	When the overall joint failure rate $p_{JF,R}$ is reasonably small\footnote{Our results show that $p_{JF,R}=0.1$ is small enough to make the conclusion, while this condition is achieved in most of the application environments.}, the logarithmic consensus failure rate ${\rm log}p_F$ has a linear relation with the faulty node threshold $f$,

\begin{equation}
\small
\begin{aligned}
&\log p_F = \\ &
\left\{
\begin{array}{ll}
k_{cf} \cdot f + h_{cf} & \text{(for CFT when $n=2f$)} \\
k_{cf} \cdot f + h_{cf} + \log(2p_{J,R}) & \text{(for CFT when $n=2f+1$)} \\
k_{bf} \cdot f + h_{bf} & \text{(for BFT when $n=3f$)} \\
k_{bf} \cdot f + h_{bf} + \log(3p_{J,R}/2) & \text{(for BFT when $n=3f+1$)} \\
k_{bf} \cdot f + h_{bf} + 2\log(3p_{J,R}/2) & \text{(for BFT when $n=3f+2$)}
\end{array}
\right.
\end{aligned}
\end{equation}

	Here $p_{J,R}=1-p_{JF,R}$. For CFT the Tolerance Gain $k_{cf}=({\rm log}p_{JF,R}+{\rm log}p_{J,R}+2{\rm log}2)$ and the intercept $h_{cf}={\rm log}(\frac{p_{JF,R}}{p_{J,R}\sqrt{\pi}})+\Delta f $, in which $\Delta f=-\frac{1}{2}{\rm log}(f)$ is the non-linear complementary term to decrease the approximation error. For BFT, the Tolerance Gain $k_{bf}={\rm log}p_{JF,R}+2{\rm log}p_{J,R}+3{\rm log}3-2{\rm log}2$ and the intercept $h_{bf}={\rm log}(\frac{\sqrt{3}p_{JF,R}}{\sqrt{\pi}p_{J,R}})+\Delta f $.
\end{theorem}
The proof of Theorem \ref{the:RA} is shown in Appendix \ref{proof of the:RA}.
  
Tolerance Gain shows that the probability of reaching consensus can be altered by not only changing $p_{JF,R}$ but also the size of the network to modify $f$. Fault tolerance allowing at most $f$ nodes failed infers one feature of the distributed consensus system is resilience. It is shown in Theorem \ref{the:RA} that $  {\rm log}p_F$ will decrease linearly with the increment of $f$. This intuitively reflects the impact of the resilience of the system on consensus reliability.

With the help of Reliability Gain and Tolerance Gain, we can quickly calculate the required link loss rate, node failure rate, or the number of nodes in the system to satisfy any stringent reliability requirement. Therefore, it is suggested that the concepts of Reliability Gain and Tolerance Gain can be the design guideline towards real-world consensus deployment with possible imperfect communication reliability

\begin{table*}[htbp]
\renewcommand\arraystretch{1.6}
\caption{\textcolor{black}{Summary of $P_{JF,R}$ for representative consensus algorithms.}}
\begin{tabular}{|c|c|c|c|c|}
\hline
Protocol & Raft \cite{ongaro2014search} & Paxos \cite{lamport2001paxos} & PBFT \cite{castro1999practical} & Hotstuff \cite{yin2019hotstuff}\\
\hline
$P_{JF,R}$ & $1-p_Np_L^2$ & $\sqrt[f+1]{2}(1-p_Np_L^2)$ & \makecell[c]{$1-p_Np_L\bar{p_*}$, \\ where $\bar{p_*}$ satisfies Eq. (\ref{mv_D_approx})} & $\sqrt[f+1]{(1-p_Np_L^2)^{f+1}+(1-p_Np_L^3)^{f+1}+2(1-p_Np_L^4)^{f+1}}$\\
\hline
\end{tabular}
\label{table: equations summary}
\end{table*}

\section{Latency and Consensus Reliability}
\label{sec: consensus reliability and latency}
Consensus latency represents the duration spanning from the initial arrival of instances at the leader to their subsequent commitment by the leader.
In this section, we will show how the Consensus Reliability proposed in Section \ref{rv_derivation} affects latency and how to mitigate the latency degradation in Figure \ref{fig: intro_fig} by decreasing the consensus failure rate.

\subsection{Higher Consensus Failure Rate Increases Latency}
\label{Relationship Between Latency and Consensus Reliability}
Latency is composed of processing latency, propagation latency, transmission latency, and queuing latency. We will show a higher consensus failure rate will greatly increase the transmission latency and queuing latency.

Let $L_{tr}$ denote the random variable of transmission latency, which refers to the time for transmitting messages within a network. Let $L_q$ denote the random variable of queuing latency, referring to the time the messages are waiting to be processed. \textcolor{black}{Assume that if a consensus instance fails due to an excessive number of link losses or node failures, the system will reattempt to achieve consensus for that specific instance until it succeeds. Let the latency of a single consensus attempt be $L_C$, which is determined by the consensus communication structure. For example, in the RAFT consensus, 
$L_C$ is approximate to the latency of two links, where one link is from the leader to followers and another is from followers to the leader.}

Since a failed instance will reattempt consensus until it succeeds, we could obtain that the transmission latency $L_{tr}$ is a random variable that follows the distribution as:
\begin{equation}
\label{L_distri}
    Pr(L_{tr}=KL_C)=P_F^{K-1}(1-P_F)
\end{equation}
where $K$ denotes the number of consensus attempts and $P_F$ is the consensus failure rate. Thus the expectation of transmission latency could be obtained as 
\begin{equation}
    E(L_{tr})=\sum_{K=1}^{\infty}KL_CPr(L_{tr}=KL_C)=\frac{L_C}{1-P_F}
\end{equation}
This directly shows a larger consensus failure rate leads a to higher transmission latency.

The queuing model usually relies on the specific protocol design and system implementation. In this case, we take RAFT as an example to analyze queuing latency. RAFT maintains a continuous log sequence with no gaps, which means the leader is only allowed to append new instances after existing ones, while moving or deleting previous instances is not allowed. This characteristic implies that RAFT requires instances with smaller indices to be committed before instances with larger indices. Consequently, RAFT's queuing latency arises from waiting for consensus reattempts of earlier instances. Assuming the arrival instances follow the Poisson flow, (the time interval between instance arrivals, $T_{arrive}$, follows exponential distribution) \cite{10.1214/aoms/1177728975}, the queuing latency of RAFT adheres to the $\bar{M}$/$\bar{G}$/1/FCFS queuing model \cite{10.1214/aoms/1177728975}, where $\bar{M}$ denotes a Poisson flow of arrival instances, $\bar{G}$ denotes a general service time distribution, 1 denotes a single server, and FCFS stands for first-come-first-served. One crucial parameter in the queuing model, the service time, can be determined as $T_{serve} = L_{tr} - L_C$ (consensus instances without reattempt are parallel and do not cause queuing delay).
According to Pollaczek-Khintchine Theorem \cite{khinchin1967mathematical}, the expectation of queuing latency is
\begin{equation}
\label{E(L_q)}
E(L_q)=\frac{Var(T_{serve})+E^2(T_{serve})}{2(E(T_{arrive})-E(T_{serve}))}
\end{equation}
\textcolor{black}{where $E(T_{serve})=\frac{P_F}{1-P_F}L_C$ and $Var(T_{serve})=\frac{P_F}{(1-P_F)^2}L_C^2$.}It could be directly obtained that a larger consensus failure rate leads to higher $E(T_{serve})$ and $Var(T_{serve})$ thus increase queuing latency.

\subsection{Optimization of Latency and Consensus Reliability}
\label{Optimization of Latency and Consensus Reliability}
In this section, we proposed two methods that utilize the theoretical analysis of Theorem \ref{th1} and Theorem \ref{the:RA} to effectively mitigate the latency degradation caused by link loss. The principle is to reduce the consensus failure rate thus decreasing  the system latency. 

\textbf{Method 1: Utilizing Tolerance Gain.}
\textcolor{black}{According to Theorem \ref{the:RA}, increasing the number of backup nodes, $n$, can effectively increase the maximum number of faulty nodes the system can tolerate, $f$, thus reduce the consensus failure rate $p_F$ (in Theorem \ref{the:RA} the coefficient $k_c$ and $k_b$ are less than 0). Then the transmission latency and queuing latency have a smaller probability distribution over larger values. Thus Method 1 is to appropriately increase the number of nodes.
Note that this method assumes sufficient communication resources
so that increasing participated nodes will not increase the latency of one single link communication.} Thus this method requires increasing communication resources
(more links) and computational resources (more nodes).

\textbf{Method 2: Utilizing Effective Communication Resource Allocation.}
Compared to Method 1, Method 2 optimizes latency by effectively allocating communication resources (such as signal transmission power) while keeping the number of nodes constant. Specifically, according to Theorem \ref{th1}, consensus reliability can be viewed as a function of the loss rate of all links. Based on the classical Signal-to-Noise-Ratio (SNR) model \cite{1374908}, we could model each wireless link loss through the outage probability of the Rayleigh channel (applicable to urban environments) as $P^{LF}=1-e^{-\frac{\gamma_{th}}{Gain*P_{tr}/P_{noise}}}$\cite{1374908}, where $P^{LF}$ is the link loss rate, $\gamma_{th}$ is the required target SNR, $P_{tr}$ is the signal transmission power, and $P_{noise}$ is the noise power. $Gain$ is the quality of the wireless communication environment and contains such parameters as the antenna gains, communication distance, and shadowing. Friis Propagation Formula could be used as an example to model $Gain$ \cite{2009wireless,1374908}. The principle of the optimization is to effectively allocate the total transmission power for different channels based on their $P_{noise}$ and $Gain$, so that the cluster has higher link success rates thus lower consensus failure rate and latency.

We will use RAFT as an example to show this method. Based on Theorem \ref{th1} and \ref{multiAas1}, the consensus failure rate of RAFT can be expressed as 
\begin{equation}
    	P^{RAFT}_F=\sum\limits_{t=f+1}^{n}(-1)^{t-f-1}\binom{t-1}{f}\sum_{ S^{JF}\subseteq\mathrm{\Omega}, |S^{JF}|=t}{\prod_{i\in S^{JF}} P_i^{JF}}
\end{equation}
where $P_i^{JF}=1-(1-P_i^{LF})^2$ and $P_i^{LF}=1-e^{-\frac{\gamma_{th}}{Gain_i*P_{tr,i}/P_{noise,i}}}$ is the link loss rate between the $i$-th node and leader. We could minimize the consensus failure rate to reduce the latency for given $\gamma_{th}$, $P_{noise,i}$, $Gain_i$, and the total transmitted power $P_{tr, total}$. Thus we have the following optimization problem:
\begin{align*}
\label{optim_raft}
&\min\quad P^{RAFT}_F\\
& \begin{array}{r@{\quad}r@{}l@{\quad}l}
s.t.&2\sum\limits_{i=1}^n P_{tr,i}&=P_{tr, total},  &i=1,2,3\ldots,n\\
 &P_{tr,i}&\geq0,  &i=1,2,3\ldots,n  \\
\end{array} .
\end{align*}
\textcolor{black}{We can solve this optimization problem by using Sequential Quadratic Programming \cite{doi:https://doi.org/10.1002/9781118723203.ch12}, a typical optimization algorithm for solving nonlinear optimization problems with constraints, to obtain the allocated power and each link loss rate.} Thus we obtain the optimized consensus failure rate and latency.

\section{Experiments}
\label{sec: experiments}

\subsection{\textcolor{black}{Results of Consensus Failure Rate and Latency}}
\label{sec: case study}

In this section, we implemented a RAFT consensus system to empirically validate our theoretical analysis concerning consensus failure rate and latency. Additionally, we assess the efficacy of the proposed latency optimization Method 1 and Method 2. Note that as highlighted in Section \ref{Optimization of Latency and Consensus Reliability}, these latency optimization techniques are not exclusively applicable to RAFT but also compatible to a broader range of consensus protocols.

We evaluated the system latency and consensus failure rate under varying conditions, including different link loss rates, number of nodes, and power allocations. Latency was measured from the instance's arrival at the leader until its commitment by the leader.
We set the arriving instances as the Poisson flow and the arrival rate as 60 per second. The single link latency is set as 0.1 s. According to Section \ref{Relationship Between Latency and Consensus Reliability}, the latency of one consensus attempt $L_C$ is approximately 0.2 s. The consensus reattempt timeout $L_{timeout}$ is set as 0.3 s. In other words, the leader will wait for 0.3 seconds to receive followers' responses and will reattempt consensus if a timeout occurs. We also set the node failure rate as 0 in this case study to show the effects of link loss only.

\begin{figure}[htbp]
  \centering 
  \includegraphics[width=0.5\textwidth, height = 5cm]{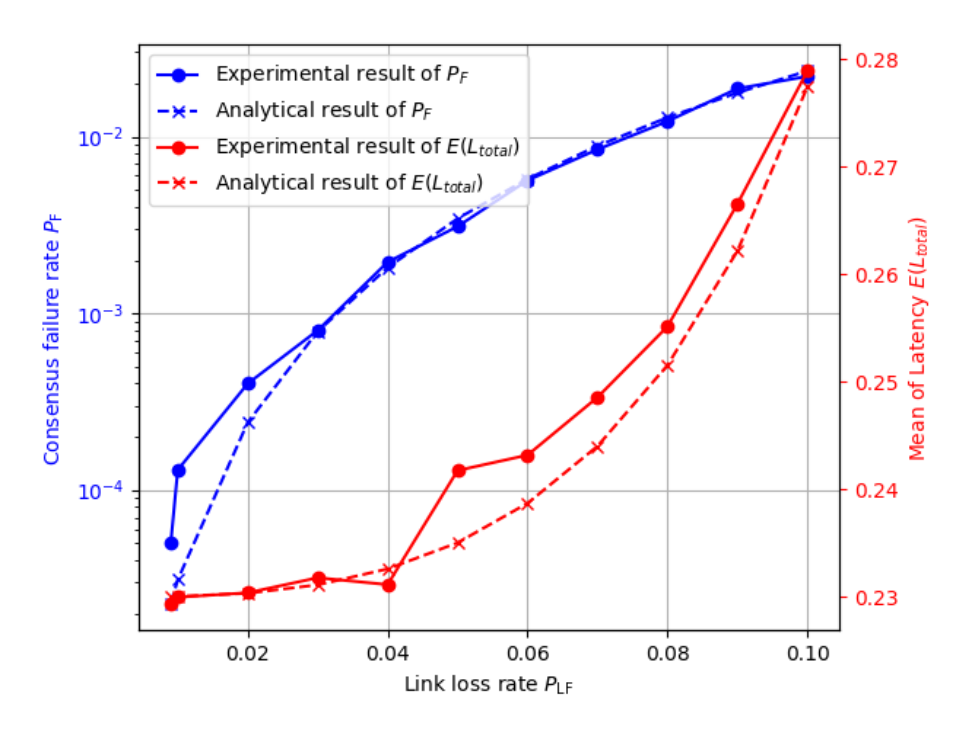}
  \caption{Analytical result and the experimental result of consensus failure rate and latency under different link loss rates when $n=4$.}
  \label{latency_failure} 
\end{figure}

\begin{figure*}
  \centering
  \includegraphics[width=1\textwidth, height = 4cm]{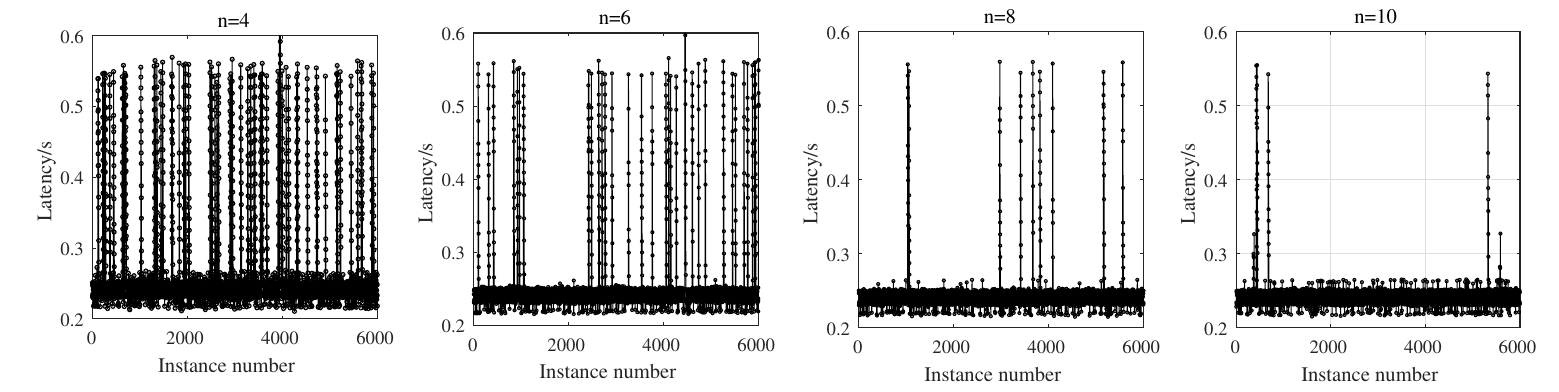}
  \caption{Latency in the RAFT system when link loss rate is 0.07 and n is 4, 6, 8, and 10.}
  \label{Appendix_fig_1} 
\end{figure*}

\begin{figure}[htbp]
  \centering 
  \includegraphics[width=7cm, height = 4cm]{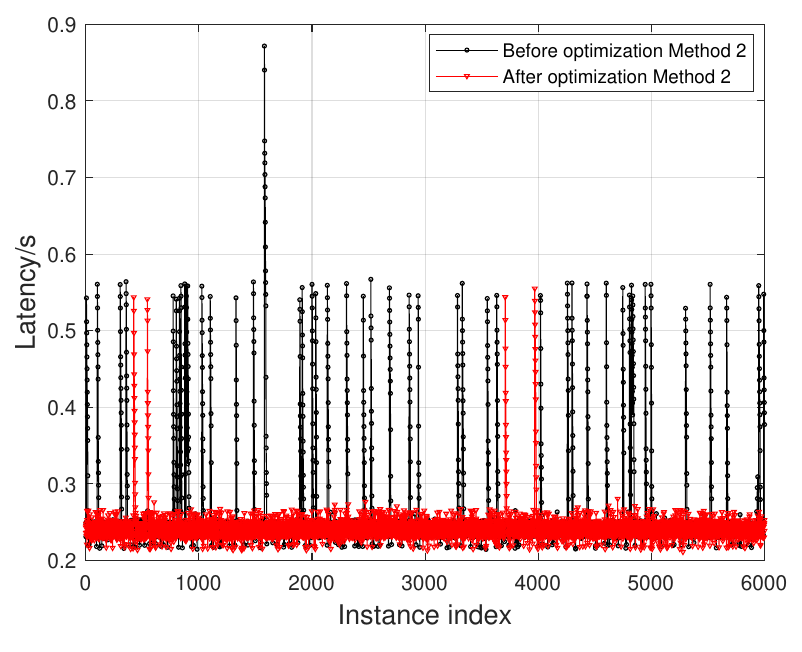}
  \caption{The comparison of latency in RAFT system before optimization Method 2 and after optimization Method 2.} 
  \label{optimization_method_2} 
\end{figure}

\textbf{Test for consensus reliability and latency:}
\label{Test for consensus reliability and latency:}

Figure \ref{latency_failure} shows the analytical result and the experimental result of consensus failure rate and latency under different link loss rates. As shown in Figure \ref{latency_failure}, the experimental result is generally consistent with the analytical result of latency and consensus failure rate.

\textbf{Test of Method 1:}  Figure \ref{Appendix_fig_1} shows the latency optimization of Method 1. This indicates that based on Tolerance Gain, increasing the number of nodes from $n=4$ to $n=10$ could decrease the order of consensus failure rate, thus improving the latency performance. The result shows that the latency degradation is mitigated and the consensus failure rate is generally consistent with the theoretical results of Tolerance Gain (i.e., less consensus failure rate for more nodes). We should mention that Method 1 requires increasing communication resources and computational resources since it needs more communication links and nodes.

\begin{figure}
\centering 
	\includegraphics[width=6cm,height = 4.5cm]{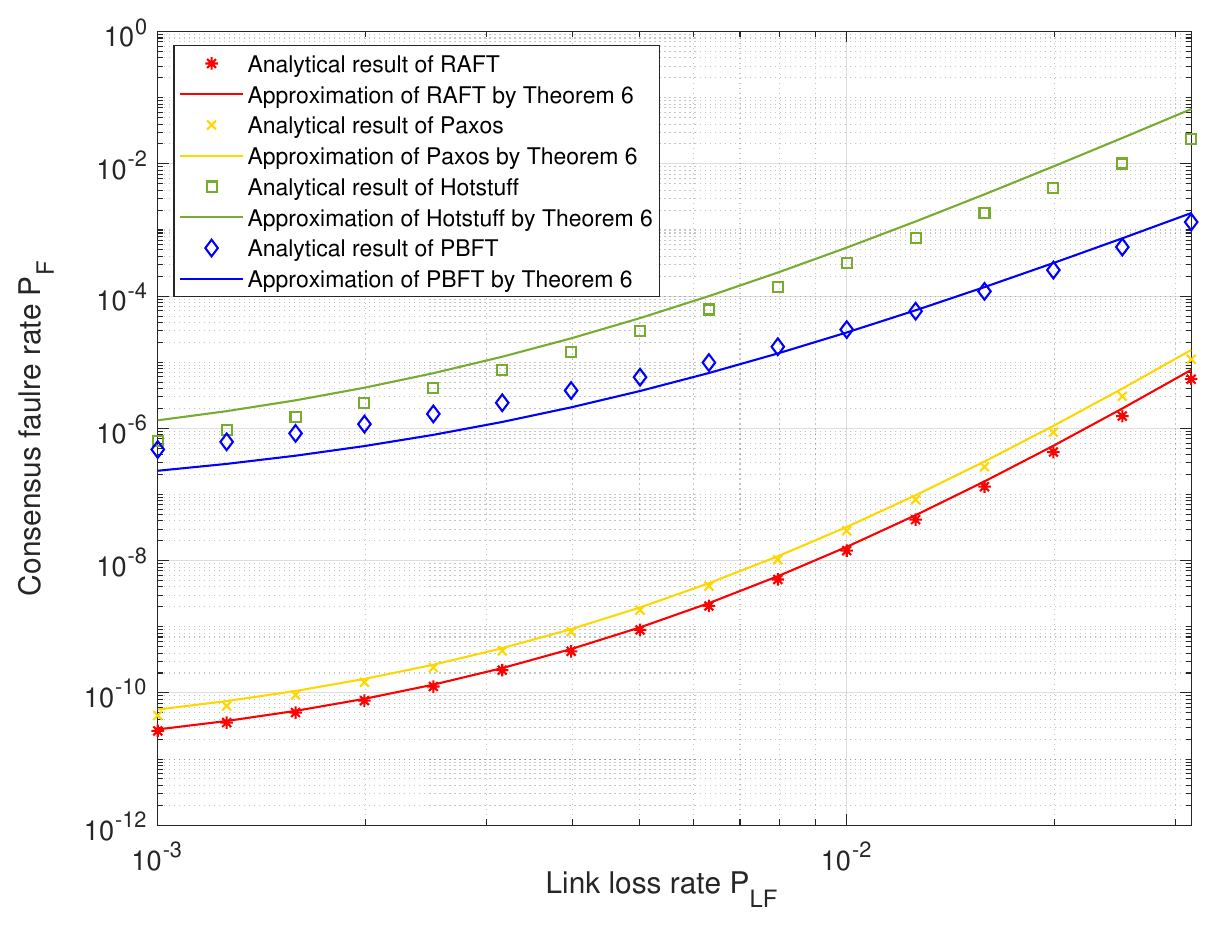}
	\caption{ Comparison between the approximate and non-approximate results of consensus failure rate $P_{F}$ for four representative consensus protocols when $n=12$ and $p_{NF}=0.01$. }
	\label{fig_approx}
\end{figure}

\begin{figure}
\centering 
	\includegraphics[width=6cm,height = 4.5cm]{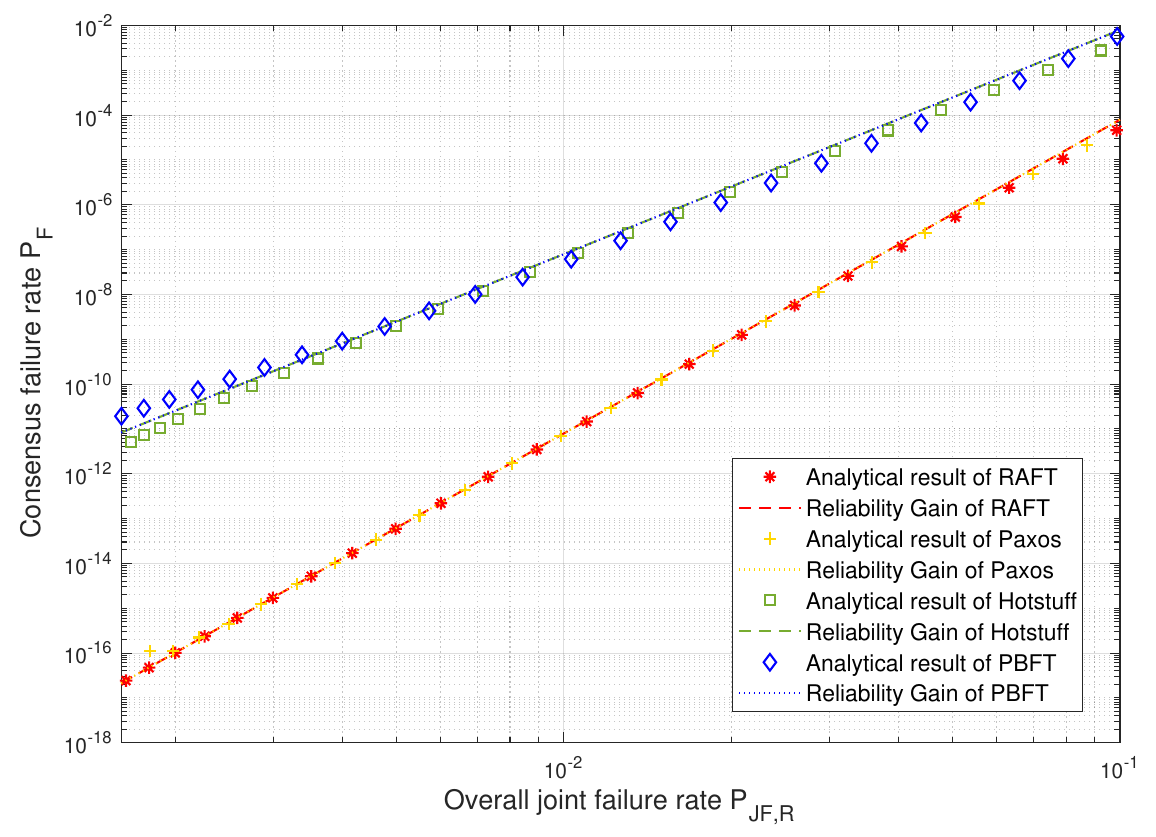}
	\caption{ Reliability Gain of consensus failure rate for four representative consensus protocols when $n=12$, $p_{NF}=0.001$, and $p_{LF}$ varies. $P_{JF,R}$ could be obtained according to Table \ref{table: equations summary}.}
	\label{fig_RG}
\end{figure}

\begin{figure}
	\centering
	\includegraphics[width=6cm,height = 4.5cm]{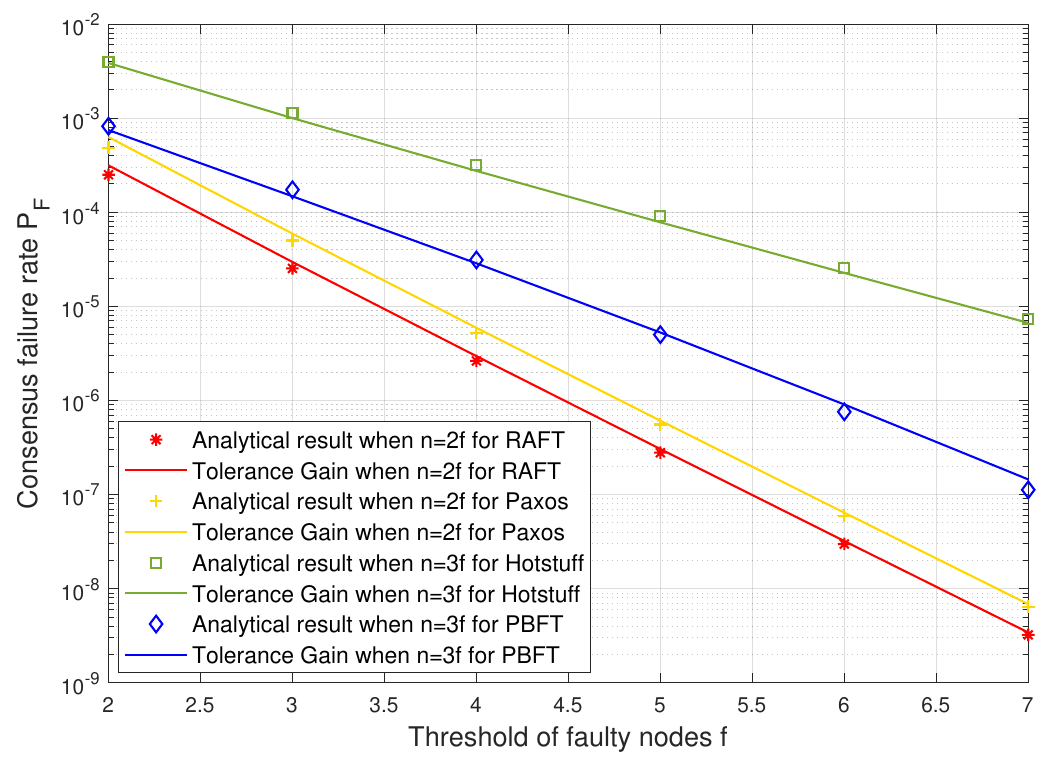}
	\caption{ Tolerance Gain of consensus failure rate when $p_{NF}=0.01$ and $p_{LF}=0.01$ for  Hotstuff and PBFT ($n=3f$), RAFT and single-decree Paxos ($n=2f$)}
	\label{fig_TG}
\end{figure}

\textbf{Test of Method 2:}
Method 2 optimizes the consensus failure rate and latency by effective power allocation without increasing communication resources while keeping the number of nodes constant. Regarding the communication model parameters, we set $n=8$, $\gamma_{th}$ as 10 dB, $P_{tr, total}$ as 1.6 W, $P_{noise}$ as $10^{-9}$ W, and  $Gain_i (dB)$ as [-68, -70, -44, -69,-65, -41, -54,-63]. \textcolor{black}{Through solving the optimization problem in Section \ref{Optimization of Latency and Consensus Reliability}, the theoretical consensus failure rates before and after optimization are 0.0081 and 0.0005, respectively}.

In Figure \ref{optimization_method_2}, the experimental results demonstrate the comparison between the latency before and after power optimization. The latency degradation caused by the link loss rate is mitigated. The experiment shows the communication power optimization method is an effective way without increasing resources in practical wireless network environments to improve the performance of consensus systems.

\textbf{One finding of Figure \ref{fig: intro_fig}, \ref{Appendix_fig_1} and \ref{optimization_method_2}}:
One finding in Figure \ref{fig: intro_fig}, \ref{Appendix_fig_1}, and \ref{optimization_method_2} is that the latency of consecutive instances experiences fluctuations, i.e., increasing then decreasing. This comes from the queuing delay in RAFT. An instance's delay abruptly increases because many links were lost and the system failed in one log replication attempt, thus it reattempts and the consensus delay doubles. Queuing occurs and the delays of subsequent instances increase because RAFT requires instances with smaller indices to be firstly committed. As time grows, the effects of queuing delay gradually decrease, and instances arriving after a period of time are no longer affected by the queuing delay of previous instances.

\subsection{Numerical Results of Reliability Gain and Tolerance Gain}
\label{Numerical Results of Consensus Failure Rate}  
Through Theorem \ref{th1} and \ref{multiAas1}, we obtain the analytical and approximation of the consensus reliability. Figure \ref{fig_approx} shows the approximate and non-approximate analytical results of consensus failure rate of four different consensus protocols when $n=12$, $p_{NF}=0.01$, and $p_{LF}$ varies. This indicates that the approximation in Theorem \ref{multiAas1} is close to the analytical results proposed in Theorem \ref{th1}, thus validates the approximations.

Reliability Gain and Tolerance Gain clearly show the linear relationship between consensus failure rate and two basic factors, $p_{JF,R}$ and $f$. We illustrate the analytical results in Figure \ref{fig_RG} and Figure \ref{fig_TG}. 
Figure \ref{fig_RG} shows the relationship between the consensus failure rate $P_F$ obtained by Theorem \ref{th1} and the overall joint failure rate $p_{JF,R}$ with $n=12$. We set $p_{NF}$ as 0.001 and change $p_{LF}$. According to Table. \ref{table: equations summary}, $p_{JF,R}$ of different protocols can be represented by $p_{NF}$ and $p_{LF}$. As shown in Figure \ref{fig_RG}, $logP_F$ obtained by Theorem \ref{th1} will be approximately linear to $logp_{JF,R}$ for different consensus protocols. This matches the Reliability Gain proposed in Theorem \ref{the:RG} (the upper lines are for BFT with $f=\lfloor n/3 \rfloor$, while the lower lines are for CFT with $f=\lfloor n/2 \rfloor$).

Theorem \ref{the:RA} demonstrates the linearity of Tolerance Gain for CFT and BFT consensus. This feature is evident in Figure \ref{fig_TG} where we set $p_{NF}=0.01$ and $p_{LF}=0.01$. In Figure \ref{fig_TG}, the consensus failure rate $P_F$ obtained by Theorem \ref{th1} for BFT consensus is approximately linear to $f$ when $n=3f$, while 
for CFT, it is approximately linear to $f$ when $n=2f$.

Figure \ref{the:RA} implies the trade-off between consensus failure rate and the tolerance to higher levels of fault. We could find that the consensus failure rate of CFT algorithms is generally smaller than that of BFT algorithms. In addition, the slopes of CFT algorithm are steeper than that of BFT algorithms, which implies the improvement of consensus reliability by increasing the faulty tolerance threshold $f$ for CFT systems is more effective than that for BFT systems.

\section{Discussion: Improve Consensus Reliability by Protocol Design}
\label{Discussion: Improve Consensus Reliability by Protocol Design}

\textcolor{black}{In this section, we discuss that increasing $r_j$ through adjusting protocol designs could exponentially decrease consensus failure rate. Recall that $r_j$ is determined by the protocol design and shows the $j$-th phase is only dependent on the $(j-r_j)$-th phase. We could obtain that a larger $r_j$ in the protocol will result in a lower consensus failure rate. This is because larger $r_j$ implies the $j$-th phase is independent with more phases between the $(j-r_j)$-th and the $j$-th phase, regardless of whether nodes in these phases have received enough messages to become activated nodes or not. 
\textcolor{black}{One example is to use this property to optimize Hotstuff. The communication structure of Hotstuff is
$G_{Hotstuff}=[A_1^{n-f},B_1^{n-f},A_2^{n-f},B_1^{n-f},A_2^{n-f},B_1^{n-f},A_2^{n-f},B_1^{f+1}]$. This shows that the nodes to be activated in each component A (one-to-many communications) are required to have successfully received messages in the previous components that are A. We could simply adjust the protocol design: use threshold signatures to incorporate replica-verifiable responses in each component that is A, which contains responses from other replicas in the previous phases. If a replica receives a message indicating that more than $n-f$ replicas have responded to the previous primary message (messages from primary nodes), even if it did not receive that previous primary message, it can proceed to the following phases in Hotstuff.
Therefore the communication structure changed as $G_{Hotstuff'}=[A_1^{n-f},B_1^{n-f},A_3^{n-f},B_1^{n-f},A_5^{n-f},B_1^{n-f},A_7^{n-f}, B_1^{f+1}]$, which implies the 1-st, 3-rd, 5-th, and 7-th phases are independent with previous components A. The nodes could be activated as long as they are non-faulty (activated in the 0-th phase) and receive the message in the component A. Let $P_{F,Hotstuff }$ be the consensus failure rate of the original Hotstuff protocol and $P_{F,Hotstuff'}$ be the consensus failure rate of the revised Hotstuff protocol. According to Corollary \ref{degen_simp}, we could obtain $\frac{P_{F,Hotstuff' }}{P_{F,Hotstuff}}\approx \frac{3*2^{f+1}+1}{2^{f+1}+3^{f+1}+2*4^{f+1}}\leq\frac{1}{2^f}$. This shows the consensus failure rate of the revised protocol is exponentially decreasing than the original one.}}

\section{Conclusion and Future Work}
\label{sec: conclusion and future work}

In this paper,  we develop a modular consensus communication framework to characterize the communication process of typical consensus protocols including RAFT, single-decree Paxos, PBFT, and Hotstuff. Through the proposed framework, we investigate the probabilistic consensus network reliability for different protocols. We also obtain that larger consensus failure rate will lead to higher latency and propose methods for latency optimization based on our theoretical results. We finally implement a RAFT system to validate our theoretical analysis and demonstrate that the latency degradation is effectively mitigated.

For future work, the proposed consensus reliability could be further explored in practical implementations of different consensus systems, such as PBFT and Hotstuff, or special observations such as the partial network partition in \cite{10.1145/3447851.3458739}. Moreover, based on our framework, the optimization of consensus failure rate could be further discussed, for example, decreasing the consensus failure rate through more effective protocol designs.

\bibliographystyle{unsrt}
	\bibliography{references}

\clearpage

\appendices
\section{Communication Structure of Representative Consensus Algorithms}
\label{apd: Communication Structure of Representative Consensus Algorithms}
To our knowledge, using the combinations of communication components introduced in Section \ref{sec: Basic Communication Components} can well represent the communication process of many consensus algorithms. As shown in Table \ref{table: popular consensus algorithm summary}, several examples of typical consensus algorithms are given including RAFT \cite{184040}, Paxos \cite{lamport2001paxos}, PBFT \cite{castro1999practical}, and Hotstuff  \cite{yin2019hotstuff}. The analysis procedures of other protocols are similar.

\subsection{Communication Structure of RAFT}
Raft \cite{184040} is a famous Crash-Fault-Tolerant consensus protocol with $f=\lfloor n/2 \rfloor$. In the log replication process of Raft, the leader sends a log entry in the one-to-many communication while followers receiving the leader's log entry will give responses in the many-to-one communication. Once received responses from the majority (i.e., more than $n-f$) of the followers, the leader decides such an entry as \textit{committed}. 
Thus the communication structure of the log replication of RAFT can be represented as $G_{Raft}=[A_{1}^{n-f},B_{1}^{n-f}]$. 

\subsection{Communication Structure of Single-decree Paxos}
Paxos is one of the most classic CFT consensus protocols. Here we do not consider Multi-Paxos which runs Paxos repeatedly and work in a state machine replication pattern, but concentrate on the single-decree Paxos proposed in \cite{lamport2001paxos}.

Paxos has four phases including Prepare, Promise, Propose, and Accept. In the Prepare-Promise phases (stage 1), the proposer sends \textit{prepare} message with number $Num$ to acceptors and receives \textit{promise} messages from the majority of acceptors. Then in the Propose-Accept phase (stage 2), the proposer tells the acceptors what value to accept through a \textit{propose} message, and the acceptors reply with \textit{accepted} message. Compared to the definition of the basic components, the Prepare-Promise phases of Paxos clearly correspond to $A_{1}^{n-f}$ followed by $B_{1}^{n-f}$. In the Propose-Accept phases, note that when an acceptor receives a \textit{propose} request with number $Num$, it only checks whether it has responded to another \textit{prepare} request having a number greater than $Num$. This means that as long as the node is non-faulty, even if it did not receive a \textit{prepare} message in stage 1, it can still respond to a \textit{propose} message in stage 2. Therefore, the \textit{propose} phase corresponds the component $A_{3}^{n-f}$. Therefore, the communication structure of Paxos is $G_{Paxos}=[A_{1}^{n-f},B_{1}^{n-f},A_{3}^{n-f},B_{1}^{n-f}]$. \textcolor{black}{Recall that in the $j$-th phase, $j-r_j=0$ implies the node is required to be AN in phase 0, i.e., to be a non-faulty node} .

\subsection{Communication Structure of PBFT}
\label{sec: structure of PBFT}

Unlike Paxos and Raft that is Crash-fault-tolerant, the Practical Byzantine Fault Tolerance (PBFT) algorithm is a Byzantine-fault-tolerant algorithm and can tolerate up to $f=\lfloor n/3\rfloor$ Byzantine faulty nodes \cite{castro1999practical}. 
There are four phases in the normal case operation of PBFT, including the \textit{pre-prepare}, \textit{prepare}, \textit{commit}, and \textit{reply} phases. In the first \textit{pre-prepare} phase, the primary sends the request from the client to replicas by a \textit{pre-prepare} message. When receiving the \textit{pre-prepare} message, replicas validate the message and accept it (turning into activated) if it is correct. Then every activated replica sends \textit{prepare} messages to other nodes. For those nodes that have been activated in the \textit{pre-prepare} phase, once receiving no less than $n-f-1$ correct prepare messages (this requires ANs in \textit{pre-prepare} phase to be more than $n-f$), a replica will turn into activated in the prepare phase. Then it will send \textit{commit} messages to other nodes. For those nodes activated in the \textit{prepare} phase, if there are no less than $n-f-1$ valid \textit{commit} messages (this requires ANs in \textit{prepare} phase should be more than $n-f$), a replica will turn into activated in the \textit{commit} phase. Finally, it will reply to the client with the result in the \textit{reply} phase. If more than $f+1$ valid \textit{reply} messages are received, the client considers the consensus achieved.
In summary, the normal case operation of PBFT can be represented as $G_{PBFT}=[A_{1}^{n-f},C_{1}^{n-f},C_{1}^{f+1},B_{1}^{f+1}]$ corresponding to \textit{pre-prepare}, \textit{prepare}, \textit{commit} and \textit{reply} phases respectively.

\subsection{Communication Structure of Hotstuff}
\label{Structure of Hotstuff}

Hotstuff \cite{yin2019hotstuff} is a leader-based Byzantine fault-tolerant protocol with linear communication complexity in the number of replicas.
Here we consider one instance of Basic Hotstuff with a stable leader skipping the first \textit{NEW-VIEW} message.
In the first Prepare phase, the leader broadcasts \textit{msg(prepare)} message to all replicas corresponding to component $A_1^{n-f}$. After validation, a replica will send \textit{votemsg(prepare)} to the leader, corresponding to component $B_1^{n-f}$. Then in the  Pre-commit phase, the leader broadcast \textit{msg(pre-commit)} message to all replicas, which will respond to the leader with \textit{votemsg(pre-commit)} once the \textit{msg(pre-commit)} message is justified. Note that a replica without receiving a valid prepare message would not respond to the following messages from the leader (due to the restriction of message type). Thus this phase should be represented as component $A_2^{n-f}$ and $B_1^{n-f}$, instead of $A_1^{n-f}$ and $B_1^{n-f}$. The Commit phase is similar to the Pre-commit phase, where components $A_2^{n-f}$ and $B_1^{n-f}$ can be used to represent the one-to-many and many-to-one communication respectively. 
In the Decide phase of Hotstuff, the leader broadcasts \textit{msg(decide)} to replicas receiving all the messages before, corresponding to the component $A_2^{n-f}$. The replicas receiving this message will execute and output this value. The consensus is considered to be achieved after the client receives replies from more than $f+1$ replicas, which corresponds to $B_1^{f+1}$. 
Therefore, the total communication of Hostuff can be represented as $G_{Hotstuff}=[A_1^{n-f},B_1^{n-f},A_2^{n-f},B_1^{n-f},A_2^{n-f},B_1^{n-f},A_2^{n-f},B_1^{f+1}]$.

\section{Proof of Theorem \ref{th1}}
\label{Proof: Theorem th1}

Let $P\left(S_{x_0}^{G_0},S_{x_1}^{G_1},...,S_{x_{\left|G\right|}}^{G_{\left|G\right|}}\right)$ denote the probability that the set of ANs in the $j$-th phase ($j=0,1,...,|G|$) is $S_{x_j}^{G_j}$. For $j$=0, consensus protocols requires the number of non-faulty nodes is no less than $n-f$, thus $n\geq x_0=|S_{x_0}^{G_0}|\geq n-f$. For $j=1,2,..|G|$, according to the basic communication component proposed in Section \ref{sec: Basic Communication Components}, $M_j$ shows the minimum number of ANs in each phase to achieve consensus. Thus we have $x_j\geq M_j$. According to Lemma \ref{rv_p_1}, we have $x_{j}\le x_{j-r_j}$ and $S_{x_j}^{G_j} \subseteq S_{x_{j-r_j}}^{G_{j-r_j}}$. Thus we could express the consensus successful rate as:
\begin{equation}
\begin{split}
    P_C=&\sum_{\substack{n\geq x_0 \geq n-f\\S_{ x_0}^{G_0}\subseteq \mathrm{\Omega}}}\!...\! \sum_{\substack{x_{j-r_j}\geq x_j \geq M_j\\S_{ x_j}^{G_j}\subseteq S_{x_{j-r_j}}^{G_{j-r_j}}}} 
\! ... \! \sum_{\substack{x_{|G|-r_{|G|}}\geq x_{|G|} \geq M_{|G|}\\S_{ x_{|G|}}^{G_{|G|}}\subseteq S_{x_{|G|-r_{|G|}}}^{G_{|G|-r_{|G|}}}}}
\\ & P\left(S_{x_0}^{G_0},S_{x_1}^{G_1},...,S_{x_{\left|G\right|}}^{G_{\left|G\right|}}\right)
\end{split}
\end{equation}
According to the chain rule and Lemma \ref{rv_p_2}, we have 
\begin{equation}
P\left(S_{x_0}^{G_0},S_{x_1}^{G_1},\ldots,S_{x_{\left|G\right|}}^{G_{\left|G\right|}}\right)=P\left(S_{x_0}^{G_0}\right)\prod_{j=1}^{\left|G\right|}P(S_{x_j}^{G_j}|S_{x_{j-r_j}}^{G_{j-r_j}})
\end{equation}
Here $P\left(S_{x_0}^{G_0}\right)$ is the probability that the nodes in $S_{x_0}^{G_0}$ are non-faulty and outside the set are faulty. $P(S_{x_j}^{G_j}|S_{x_{j-r_j}}^{G_{j-r_j}})$ is the probability that the nodes in $S_{x_j}^{G_j}$ are activated and in $\complement_{S_{x_{j-r_j}}^{G_{j-r_j}}}S_{x_j} ^{G_j}$ are not activated, We could express them as
\begin{equation}
\begin{split}
P\left(S_{x_0}^{G_0}\right)=&Pr((\bigcap_{u\in S_{x_0}^{G_0}}{u\ is\ non-faulty})\bigcap \\&{(\bigcap_{ v\in\complement_\mathrm{\Omega}S_{x_0}^{G_0}}{\ v\ is\ faulty}}))
\end{split}
\end{equation}
\begin{equation}
\begin{split}
P\left(S_{x_j}^{G_j}\middle| S_{x_{j-r_j}}^{G_{j-r_j}}\right)=&Pr((\bigcap_{u\in S_{x_j} ^{G_j}}{\ u\ is\ activated})\bigcap\\
&{(\bigcap_{v\in\complement_{S_{x_{j-r_j}}^{G_{j-r_j}}}S_{x_j} ^{G_j}}{\ v\ is\ not\ activated}}))
\end{split}
\end{equation}
Considering each node being activated in different phases are independent, we have Eq.(\ref{rv_1A_eq_3})-(\ref{rv_1A_eq_4}). Thus Theorem \ref{th1} is proved.

\section{Consensus Reliability for Multiple Instances}
\label{consecutive instances}

We have considered the consensus reliability of a single instance in Section \ref{Markov Property of Consensus} and \ref{appro_sim}. Assuming different instances are independent, we show that the consensus reliability of consecutive multiple instances can be extended 
based on that of a single instance. 

Note that link failures and node failures may have different influences for multiple instances. 
A link failure is transient, while a node failure is durative since faulty nodes are assumed not to be recovered in a sufficiently long time after the failure. Thus it is more catastrophic for node failure in consecutive instances. 

We obtain the consensus reliability of consecutive $W$ instances $P_C^{Mul}$ as follow:
\begin{theorem}
    \label{theorem:consecutive instances}
    The consensus success rate of consecutive $W$ instances can be obtained as:
    \begin{equation}
    	P_C^{Mul}\approx WP_C-(W-1)P_C^{Node}
    \end{equation}
    where $P_C$ is the consensus reliability of a single instance, and $P_C^{Node}$ is the consensus reliability of a single instance only considering the possible node failures (i.e., all the link failure rates are 0). The approximation will be tighter if the link failure rate is closer to 0.
\end{theorem}
\begin{proof}
\label{proof_theorem:consecutive instances}

We can first consider the conditional probability of single consensus instance for a given set of non-faulty nodes, $P_C|_{S_{x_0}^{G_0}}$ and consecutive $W$ instances are successful (i.e., $(P_C|_{S_{x_0}^{G_0}})^W$).  Then we use full-partition formula to obtain the  consensus reliability of consecutive instances as follow. 
\begin{equation}
\label{pCMul_1}
    P_C^{Mul}= \sum_{n\geq x_0\geq n-f}\sum_{ \mathrm{\Omega}\supseteq S_{x_0}^{G_0}}{P(S_{x_0}^{G_0})}(P_C|_{S_{x_0}^{G_0}})^W
\end{equation}

Since the link failure rate is always closer to 0, $1-P_C|_{S_{x_0}^{G_0}}$ is close to 0. We have
\begin{equation}
	\begin{split}
	\label{proof1_2}
	(P_C|_{S_{x_0}^{G_0}})^W=(1-(1-P_C|_{S_{x_0}^{G_0}}))^W\approx WP_C|_{S_{x_0}^{G_0}}-(W-1)
	\end{split}
\end{equation}

By using Eq. (\ref{proof1_2}), Eq. (\ref{pCMul_1}) can be further transformed as
\begin{equation}
	\begin{split}
		P_C^{Mul}&\approx \sum_{n\geq x_0\geq n-f}\sum_{ \mathrm{\Omega}\supseteq S_{x_0}^{G_0}}{P(S_{x_0}^{G_0})}(WP_C|_{S_{x_0}^{G_0}}-(W-1))\\&= WP_C-(W-1)\sum_{n\geq x_0\geq n-f}\sum_{ \mathrm{\Omega}\supseteq S_{x_0}^{G_0}}{P(S_{x_0}^{G_0})}\\&= WP_C-(W-1)P_C^{Node}
	\end{split}
\end{equation}
where $P_C$ is the consensus reliability of a single instance, and $P_C^{Node}$ is the consensus reliability of a single instance only considering the possible node failures (i.e., all the link failure rates are 0). The approximation will be tighter if the link failure rate is closer to 0. Thus Theorem \ref{theorem:consecutive instances} has been proved.
\end{proof}

Based on the consensus reliability of a single instance, the consensus reliability of consecutive $W$ instances can be obtained through Theorem \ref{theorem:consecutive instances}. It can be approximately evaluated that how many consecutive instances will cause at least one consensus instance failure.

\section{Proof of Theorem \ref{multiAas1}}
\label{proof of multiAas1}

We first propose the methods of joint failure and power series to provide a flexible way of simplifying the consensus reliability for a special case. Then we extend the special case to generic cases to prove Theorem \ref{proof of multiAas1}.

\subsection{Part 1: Method of Joint Failure}
\begin{lemma} Method of Joint Reliability:
\label{rv_th_jfv}
If $M_j=n-f$ and $r_j=1$ for any $j=1,2,...|G|$, the consensus reliability, $P_C$, can be simplified as:
\begin{equation}
\label{as1}
P_C\approx \sum_{\substack{n\geq x\geq n-f\\ \mathrm{\Omega}\supseteq S_x}} P\left(S_x\right) =\sum_{\substack{n\geq x\geq n-f\\ \mathrm{\Omega}\supseteq S_x}}{\prod_{u\in S_x} P_u^J\prod_{v\in\complement_\mathrm{\Omega}S_x}\left( 1-P_v^J\right)}
\end{equation}
where $S_x$ is the running variable with $|S_x|=x$ representing the set of nodes always being activated from phase 0 to the final phase, and 
\begin{equation}
\begin{split}
\label{as2}
P_i^J=\prod_{j=0}^{\left|G\right|}P_i^{G_j}
\end{split}
\end{equation}
is the activation probability of node $i$ from phase 0 to the final phase (joint reliability).

\end{lemma}
Here $r_j=1$ for $j=1,2,...|G|$ implies the first-order Markov property is held. Please see the proof in Appendix \ref{proof of joint failure}.

\textcolor{black}{When there is no graph C in the communication structure of the consensus protocol $G$, $P_C$ in Eq. (\ref{as1}) is not approximated since the derivations in Appendix \ref{proof of joint failure} are identity transformations. 
However, when the consensus protocol has graph C in the $j$-th phase, according to Eq. (\ref{rv_1A_eq_5}), $P_i^{G_j}$ is determined by the set $S_{x_{j-r_j}}^{G_{j-r_j}}$, which will vary in different summation paths. 
When the graph of $j$-th phase is C, we use the approximation form of Eq. (\ref{rv_1A_eq_5}), and the approximation will be tight if the link reliability is close to 1,}
\begin{equation}
\begin{split}
\label{as3}
\bar{P}_i^{G_j}= \sqrt[f+1]{\frac{\sum_{w=n-f-1}^{n-1}{(\bar{P}_i^{G_j}(w))^{f+1}}}{f+1}} 
\end{split}
\end{equation}
where $\bar{P}_i^{G_j}(w)$ is
\begin{equation}
\begin{split}
\label{as4}
\bar{P}_i^{G_j}(w)=\frac{\sum_{\Phi_{w}\in (\Omega\setminus \{i\})}{P_i^{G_j}(\{ \Phi_w, i\})} }{\binom{n-1}{w}}
\end{split}
\end{equation}
Here $\Phi_{w}$ is a running variable with $|\Phi_w|=w$, $P_i^{G_j}(\{ \Phi_w, i\})$ is the left side in Eq. (\ref{rv_1A_eq_5}) replacing $S_{x_{j-r_j}}^{G_{j-r_j}}$ as $\{ \Phi_w, i\}$ .

\subsection{Part 2: Method of Power Series}
We show the second simplification method as follows.

\begin{lemma}  Power Series Method (\cite{Li}):
\label{rv_th_ps}
If $M_j=n-f$ and $r_j=1$ for any $j=1,2,...|G|$, the consensus failure rate, denoted as $P_F=1-P_C$, can be simplified as:
    \begin{equation}
    	\label{RV:PS}
    	P_F=\sum\limits_{t=f+1}^{n}a_{t}Q_{t}
    \end{equation}
    where $Q_{t}=\sum_{\mathrm{\Omega}\supseteq S^{JF}_{t}}{\prod_{u\in S^{JF}_{t}} P_u^{JF}}$
    is the summation of the product of $t$ joint failure rates, which can be considered as summation of power of $t$, $S^{JF}_{t}$ is a running variable with $|S^{JF}_{t}|=t$, $P_u^{JF}=1-P_u^{J}$ is the joint failure rate, and $a_{t}=(-1)^{t-f-1}\binom{t-1}{f}$ is the coefficient of the series expansion.
\end{lemma}

Lemma \ref{rv_th_ps} takes the view of failure probability. Particularly, $Q_{t}$ can be considered as summation of power of $t$, since the term $\prod_{u\in S^{JF}_{t}} P_u^{JF}$ is the product of $t$ joint failure rates. In practical application scenarios, the joint failure rate $P^{JF}_{u}$ is usually small and close to $0$. Otherwise, a large number of nodes and links would fail so that the system might not carry out any consensus instance. 
Taking advantage of this characteristic, the term $Q_{t}$ will become smaller with the increment of $t$. 
Thus, $Q_{t}$ with high-order power $t$ could be omitted to achieve the purpose of simplifying $P_F$. This provides a flexible way to simplify $P_F$ according to actual approximate accuracy requirements. 
\begin{corollary}
\label{oneAsimp}
After only retaining the first non-zero term of Lemma \ref{rv_th_ps}, we have
  \begin{equation}
    	\label{RV:simp}
    	P_{F}\approx Q_{f+1}=\sum_{\mathrm{\Omega}\supseteq S^{JF}_{f+1}}{\prod_{u\in S^{JF}_{f+1}} P_u^{JF}} 
    \end{equation}
\end{corollary}

\subsection{Part 3: Proof of Theorem \ref{multiAas1}}
Then based on Lemma \ref{rv_th_jfv} and \ref{rv_th_ps}, we could prove Theorem \ref{multiAas1}. 

\begin{theorem}
Given the communication structure $G$, consider all the dependence relationships in $G$ as a graph where each phase is a node and each phase dependence relationship $r_j$ is an edge. All the dependence relationships form a tree, where the 0-th phase is the root node. Assume that there are $L$ leaf nodes and let $R_l(l=1,2,...,L)$ denote the set of phases in the path between the root node and the $l$-th leaf node, the consensus failure rate could be approximated as:
\begin{equation}
\label{multiAas1_1}
P_F\approx\sum_{l=1}^{L}P_F^{R_l}
\end{equation}
where $P_F^{R_l}$ is
\begin{equation}
\label{multiAas2}
P_F^{R_l} = \sum_{\mathrm{\Omega}\supseteq S_{f+1}^{JF,R_l}}\prod_{u\in S_{f+1}^{R_l}} P_u^{ R_l}
\end{equation}
where $P_u^{JF,R_l}=1-\prod_{k \in R_l}P_u^{G_k}$ and $P_u^{G_k}$ is the probability that $u$ becoming AN in the $k$-th phase.  
The approximation will be tight if the node/ link failure rates are close to 1. 
\end{theorem}

If we consider all the dependence relationships in $G$ as a graph where each phase is a node and each phase dependence relationship $r_j$ is an edge, all the dependence relationships will form a tree, where the $0$-th phase is the root node. 

In this proof, we first consider an approximation for a basic forked communication structure, then show that all the forked phases in the tree could use such a trick and obtain Theorem \ref{multiAas1}.

Consider the following basic forked communication structure $G_{forked}$. Let $a,b,c$ be indexes of three phases with $0\leq a\leq b \leq c$, where the $b$-th phase is dependent on the $a$-th phase and other phases follow the first-order Markov property. This means when $j=1,2,...,b-1$, $r_j=1$; when $j=b$, $r_j=b-a$; when $j=b+1,b+2,...,c$, $r_j=1$. We call the $a$-th phase a forked phase since both the $(a+1)$-th phase and the $b$-th phase are dependent on it. According to Theorem \ref{th1} and Lemma \ref{rv_th_jfv}, we have
\begin{equation}
\small
\label{proof_approx_theo_1}
\begin{split}
&P_C\approx \sum_{\substack{n\geq x_0 \geq n-f\\S_{ x_0}^{0->a}\subseteq \mathrm{\Omega}}}P(S_{ x_0}^{0->a})\\ &\sum_{\substack{x_{0}\geq x_1 \geq n-f\\S_{ x_1}^{a+1->b-1}\subseteq S_{ x_0}^{0->a}}}\sum_{\substack{x_{0}\geq x_2 \geq n-f\\S_{ x_2}^{b->c}\subseteq S_{ x_0}^{0->a}}} \\ & P(S_{ x_1}^{a+1->b-1}|S_{ x_0}^{0->a})P(S_{ x_2}^{b->c}|S_{ x_0}^{0->a})\\ &=\sum_{\substack{n\geq x_0 \geq n-f\\S_{ x_0}^{0->a}\subseteq \mathrm{\Omega}}}P(S_{ x_0}^{0->a})(\sum_{\substack{x_{0}\geq x_1 \geq n-f\\S_{ x_1}^{a+1->b-1}\subseteq S_{ x_0}^{0->a}}}P(S_{ x_1}^{a+1->b-1}|S_{ x_0}^{0->a})) \\ &(\sum_{\substack{x_{0}\geq x_2 \geq n-f\\S_{ x_2}^{b->c}\subseteq S_{ x_0}^{0->a}}}P(S_{ x_2}^{b->c}|S_{ x_0}^{0->a}))
\end{split}
\end{equation}
where $P(S_{ x_0}^{0->a}),P(S_{ x_1}^{a+1->b-1}|S_{ x_0}^{0->a})$ and $P(S_{ x_2}^{b->c}|S_{ x_0}^{0->a})$ can be obtained according to Eq.(\ref{rv_1A_eq_3}) and (\ref{rv_1A_eq_4}).
Note here by using Lemma \ref{rv_th_jfv}, $P_i^{0->a}=\prod_{j=0}^{a} P_i^{G_j}$, $P_i^{a+1->b-1}=\prod_{j=a+1}^{b-1} P_i^{G_j}$, and $P_i^{b->c}=\prod_{j=b}^{c} P_i^{G_j}$ are the activation probability of node $i$ in the phases from $0$ to $a$, from $a+1$ to $b-1$, and from $b$ to $c$, respectively. 

Let $A=\sum_{\substack{x_{0}\geq x_1 \geq n-f\\S_{ x_1}^{a+1->b-1}\subseteq S_{ x_0}^{0->a}}}P(S_{ x_1}^{a+1->b-1}|S_{ x_0}^{0->a})$ and \\ $B=\sum_{\substack{x_{0}\geq x_2 \geq n-f\\S_{ x_2}^{b->c}\subseteq S_{ x_0}^{0->a}}}P(S_{ x_2}^{b->c}|S_{ x_0}^{0->a})$.
Since the link failure rate is always close to 0, in Eq.(\ref{proof_4_1}) A and B are close to 1.  Thus we could obtain the approximation:
\begin{equation}
	\begin{split}
	\label{proof_approx_theo_2}
 AB&=(1-(1-A))(1-(1-B)) \\ & \approx 1-(1-A)-(1-B)=A+B-1
	\end{split}
\end{equation}
Thus Eq. (\ref{proof_approx_theo_1}) could be transformed as:
\begin{equation}
\small
\label{proof_approx_theo_3}
\begin{split}
&P_C\approx \sum_{\substack{n\geq x_0 \geq n-f\\S_{ x_0}^{0->a}\subseteq \mathrm{\Omega}}}\sum_{\substack{x_{0}\geq x_1 \geq n-f\\S_{ x_1}^{a+1->b-1}\subseteq S_{ x_0}^{0->a}}}P(S_{ x_0}^{0->a})P(S_{ x_1}^{a+1->b-1}|S_{ x_0}^{0->a})\\ &+\sum_{\substack{n\geq x_0 \geq n-f\\S_{ x_0}^{0->a}\subseteq \mathrm{\Omega}}}\sum_{\substack{x_{0}\geq x_2 \geq n-f\\S_{ x_2}^{b->c}\subseteq S_{ x_0}^{0->a}}}P(S_{ x_0}^{0->a})P(S_{ x_2}^{b->c}|S_{ x_0}^{0->a})\\ & -\!\!\!\!\sum_{\substack{n\geq x_0 \geq n-f\\S_{ x_0}^{0->a}\subseteq \mathrm{\Omega}}}P(S_{ x_0}^{0->a})\\ &=
\sum_{\substack{n\geq x_1 \geq n-f\\S_{ x_1}^{0->b-1}\subseteq \mathrm{\Omega}}}P(S_{ x_0}^{0->b-1})+\sum_{\substack{n\geq x_2 \geq n-f\\S_{ x_2}^{0->a,b->c}\subseteq \mathrm{\Omega}}}P(S_{ x_2}^{0->a,b->c})\\ &-\sum_{\substack{n\geq x_0 \geq n-f\\S_{ x_0}^{0->a}\subseteq \mathrm{\Omega}}}P(S_{ x_0}^{0->a})\\ &\approx 
\sum_{\substack{n\geq x_1 \geq n-f\\S_{ x_1}^{0->b-1}\subseteq \mathrm{\Omega}}}P(S_{ x_0}^{0->b-1})+\sum_{\substack{n\geq x_2 \geq n-f\\S_{ x_2}^{0->a,b->c}\subseteq \mathrm{\Omega}}}P(S_{ x_2}^{0->a,b->c})-1
\end{split}
\end{equation}
where for $P(S_{ x_0}^{0->b-1})$ and $P(S_{ x_2}^{0->a,b->c})$, we have the corresponding $P_i^{0->b-1}=\prod_{j=0}^{b-1} P_i^{G_j}$ and $P_i^{0->a,b->c}=(\prod_{j=0}^{a}P_i^{G_j})(\prod_{j=b}^{c} P_i^{G_j})$. 

Thus we could obtain the consensus failure rate as
\begin{equation}
\label{proof_approx_theo_4}
\begin{split}
&P_F=1-P_C\\& \approx
(1-\sum_{\substack{n\geq x_1 \geq n-f\\S_{ x_1}^{0->b-1}\subseteq \mathrm{\Omega}}}P(S_{ x_0}^{0->b-1})) \\ & +(1-\sum_{\substack{n\geq x_2 \geq n-f\\S_{ x_2}^{0->a,b->c}\subseteq \mathrm{\Omega}}}P(S_{ x_2}^{0->a,b->c}))
\end{split}
\end{equation}
Recall that all the dependence relationships form a tree, where the $0$-th phase is the root node. We could find that in Eq. (\ref{proof_approx_theo_4}), $0-1-2-,...,-(b-2)-(b-1)$ is exactly the path from the $0$-th phase (the root node) to the $(b-1)$-th phase (one leaf node) and $0-1-...-(a-1)-a-b-(b+1)-...-c$ is exactly the path from the $0$-th phase (the root node) to the $c$-th phase (another leaf node). 

Eq. (\ref{proof_approx_theo_1})-(\ref{proof_approx_theo_4}) show the approximation of a forked phase, i.e., the $a$-th phase. Regarding the tree formed by the dependence relationships in the communication structure $G$ for some protocol, we conduct similar operations from Eq. (\ref{proof_approx_theo_1}) to Eq. (\ref{proof_approx_theo_4}) for all the forked phases, thus we could obtain Eq. (\ref{multiAas1_1}).  By Corollary \ref{oneAsimp},  we obtain Eq. (\ref{multiAas2}) for each $P_F^{R_l}$,. The approximation will be tight if the activation probabilities are close to 1, which implies node/ link failure rates are close to 1. Thus we have proven Theorem \ref{multiAas1}.

\clearpage
\begin{strip}
\subsection{Part 4: Proof of Lemma \ref{rv_th_jfv}}
\label{proof of joint failure}
If $r_j=1$ and $M_j=n-f$ for any $j=1,...,|G|$, Eq. (\ref{eq_th1}) is transformed as: 
\begin{equation}
\label{proof_th1_1}
P_C=\sum_{\substack{n\geq x_0\geq x_1\geq ...\geq x_{\left|G\right|}\geq n-f\\ \mathrm{\Omega}\supseteq S_{ x_0}^{G_0}\supseteq S_{x_1}^{G_1}...\supseteq S_{x_{\left|G\right|}}^{G_{\left|G\right|}}}} P\left(S_{x_0}^{G_0},S_{x_1}^{G_1},\ldots,S_{x_{\left|G\right|}}^{G_{\left|G\right|}}\right)
\end{equation}
where 
\begin{equation}
\label{proof_th1_2}
P\left(S_{x_0}^{G_0},S_{x_1}^{G_1},\ldots,S_{x_{\left|G\right|}}^{G_{\left|G\right|}}\right)=P\left(S_{x_0}^{G_0}\right)\prod_{j=1}^{\left|G\right|}P(S_{x_j}^{G_j}|S_{x_{j-1}}^{G_{j-1}})
\end{equation}
\begin{equation}
\label{proof_th1_3}
P\left(S_{x_0}^{G_0}\right)=\prod_{u\in S_{x_0}^{G_0}} P_u^{G_0}\prod_{v\in\complement_\mathrm{\Omega} S_{x_0}^{G_0}}\left(1-P_v^{G_0}\right), \ \ 
P\left(S_{x_j}^{G_j}\middle| S_{x_{j-1}}^{G_{j-1}}\right)=\prod_{u\in S_{x_j}^{G_j }} P_u^{G_j}\prod_{v\in\complement_{S_{x_{j-1}}^{G_{j-1}}}S_{x_j}^{G_j}}\left(1-P_v ^{G_j}\right)
\end{equation}
By using Eq. (\ref{proof_th1_2}) and changing the summation order in Eq. (\ref{proof_th1_1}), we have:
\begin{equation}
\label{proof_2_1}
\begin{split}
P_C&=\sum_{\substack{n\geq x_0\geq x_1\geq ...\geq x_{\left|G\right|}\geq n-f\\ \mathrm{\Omega}\supseteq S_{ x_0}^{G_0}\supseteq S_{x_1}^{G_1}...\supseteq S_{x_{\left|G\right|}}^{G_{\left|G\right|}}}}P(S_{x_0}^{G_0})\prod_{j=1}^{|G|}P(S_{x_j}^{G_j}|S_{x_{j-1}}^{G_{j-1}})=\sum_{\substack{n\geq x_1\geq ...\geq x_{\left|G\right|}\geq n-f\\ \mathrm{\Omega}\supseteq S_{ x_1}^{G_1} ...\supseteq S_{x_{|G|}}^{G_{|G|}}}}\prod_{j=2}^{|G|}P(S_{x_j}^{G_j}|S_{x_{j-1}}^{G_{j-1}}) \sum_{\substack{n\geq x_0\geq x_1\\ \mathrm{\Omega}\supseteq S_{ x_0}^{G_0}\supseteq S_{x_1}^{G_1}}}P(S_{x_0}^{G_0})P(S_{x_1}^{G_1}|S_{x_0}^{G_0}) 
\end{split}
\end{equation}

Substitute Eq. (\ref{proof_th1_3}) into Eq. (\ref{proof_2_1}):
\begin{equation}
\label{proof_2_2}
\begin{split}
P_C&=\sum_{\substack{n\geq x_1\geq ...\geq x_{\left|G\right|}\geq n-f\\ \mathrm{\Omega}\supseteq S_{ x_1}^{G_1} \supseteq ... S_{x_{|G|}}^{G_{|G|}}}}\prod_{j=2}^{|G|}P(S_{x_j}^{G_j}|S_{x_{j-1}}^{G_{j-1}}) \sum_{\substack{n\geq x_0\geq x_1\\ \mathrm{\Omega}\supseteq S_{ x_0}^{G_0}\supseteq S_{x_1}^{G_1}}}\prod_{u_0\in S_{x_0}^{G_0 }} P_{u_0}^{G_0}\prod_{v_0\in\complement_{\Omega}S_{x_0}^{G_0}}\left(1-P_{v_0} ^{G_0}\right)\prod_{u_1\in S_{x_1}^{G_1 }} P_{u_1}^{G_1}\prod_{v_1\in\complement_{S_{x_{0}}^{G_{0}}}S_{x_1}^{G_1}}\left(1-P_{v_1} ^{G_1}\right)
\end{split}
\end{equation}

Considering the last summation, by using $\prod_{u_0\in S_{x_0}^{G_0 }} P_{u_0}^{G_0}=\prod_{u_1\in S_{x_1}^{G_1}} P_{u_1}^{G_0}\prod_{v_1\in \complement_{S_{x_0}^{G_0}}S_{x_1}^{G_1}} P_{v_1}^{G_0}$ we have
\begin{equation}
\label{proof_2_3}
\begin{split}
&\sum_{\substack{n\geq x_0\geq x_1\\ \mathrm{\Omega}\supseteq S_{ x_0}^{G_0}\supseteq S_{x_1}^{G_1}}}\prod_{u_0\in S_{x_0}^{G_0 }} P_{u_0}^{G_0}\prod_{v_0\in\complement_{\Omega}S_{x_0}^{G_0}}\left(1-P_{v_0} ^{G_0}\right)\prod_{u_1\in S_{x_1}^{G_1 }} P_{u_1}^{G_1}\prod_{v_1\in\complement_{S_{x_{0}}^{G_{0}}}S_{x_1}^{G_1}}\left(1-P_{v_1} ^{G_1}\right)\\&=\sum_{\substack{n\geq x_0\geq x_1\\ \mathrm{\Omega}\supseteq S_{ x_0}^{G_0}\supseteq S_{x_1}^{G_1}}}{\prod_{v_0\in\complement_\mathrm{\Omega}S_{x_0}^{G_0}}{(1-P_{v_0}^{G_0})} \prod_{u_1\in S_{x_1}^{G_1}}{P_{u_1}^{G_{0}}P_{u_1}^{G_1}\prod_{v_1\in\complement_{S_{x_0}^{G_0}}S_{x_1}^{G_1}}{P_{v_1}^{G_0}(1-P_{v_1}^{G_1}}})}
\\ & =\prod_{u_1\in S_{x_1}^{G_1}}{P_{u_1}^{G_0}P_{u_1}^{G_1}}(\sum_{\substack{n\geq x_0\geq x_1\\ \mathrm{\Omega}\supseteq S_{ x_0}^{G_0}\supseteq S_{x_1}^{G_1}}}\prod_{v_0\in\complement_\mathrm{\Omega}S_{x_0}^{G_0}}{(1-P_{v_0}^{G_0})}\prod_{v_1\in\complement_{S_{x_0}^{G_0}}S_{x_1}^{G_1}}{P_{v_1}^{G_0}(1-P_{v_1}^{G_1})})
\end{split}
\end{equation}

Since $\complement_{\mathrm{\Omega}}S_{x_1}^{G_1}=\complement_{\mathrm{\Omega}}S_{x_0}^{G_0}\cup\complement_{S_{x_0}^{G_0}}S_{x_1}^{G_1},\ \ \complement_{\mathrm{\Omega}}S_{x_0}^{G_0}\cap\complement_{S_{x_0}^{G_0}}S_{x_1}^{G_1}=\emptyset$ and the summation traversed all sets from $\mathrm{\Omega}$ to $S_{x_1}^{G_1}$, thus
\begin{equation}
\begin{split}
\label{proof_2_4}
&\prod_{u_1\in S_{x_1}^{G_1}}{P_{u_1}^{G_0}P_{u_1}^{G_1}}(\sum_{\substack{n\geq x_0\geq x_1\\ \mathrm{\Omega}\supseteq S_{x_0}^{G_0}\supseteq S_{x_1}^{G_1}}}\prod_{v_0\in\complement_\mathrm{\Omega}S_{x_0}^{G_0}}{(1-P_{v_0}^{G_0})}\prod_{v_1\in\complement_{S_{x_0}^{G_0}}S_{x_1}^{G_1}}{P_{v_1}^{G_0}(1-P_{v_1}^{G_1})})\\&=\prod_{u_1\in S_{x_1}^{G_1}}{P_{u_1}^{G_0}P_{u_1}^{G_1}}\prod_{ v_1\in\complement_\mathrm{\Omega}S_{x_1}^{G_1}}{((1-P_{v_1}^{G_0})+P_{v_1}^{G_0}(1-P_{v_1}^{G_1}))}=\prod_{u_1\in S_{x_1}^{G_1}}{P_{u_1}^{G_0}P_{u_1}^{G_1}}\prod_{v_1\in\complement_\mathrm{\Omega}S_{x_1}^{G_1}}{(1-{P}_{v_1}^{G_0}P_{v_1}^{G_{1}})}
\end{split}
\end{equation}

\clearpage
Thus substituting Eq. (\ref{proof_2_3}) and Eq. (\ref{proof_2_4}) into Eq. (\ref{proof_2_2}), $P_C$ can be transformed as:

\begin{equation}
\begin{split}
\label{proof_2_5}
P_C&=\sum_{\substack{n\geq x_1\geq ...\geq x_{\left|G\right|}\geq n-f\\ \mathrm{\Omega}\supseteq S_{ x_1}^{G_1} \supseteq S_{x_{|G|}}^{G_{|G|}}}}\prod_{j=2}^{|G|}P(S_{x_j}^{G_j}|S_{x_{j-1}}^{G_{j-1}})(\prod_{u_1\in S_{x_1}^{G_1}}{P_{u_1}^{G_0}P_{u_1}^{G_1}}\prod_{v_1\in\complement_\mathrm{\Omega}S_{x_1}^{G_1}}{(1-{P}_{v_1}^{G_0}P_{v_1}^{G_{1}})}) \\ &=\sum_{\substack{n\geq x_2...\geq x_{\left|G\right|}\geq n-f\\ \mathrm{\Omega}\supseteq S_{ x_2}^{G_2} \supseteq S_{x_{|G|}}^{G_{|G|}}}}\prod_{j=3}^{|G|}P(S_{x_j}^{G_j}|S_{x_{j-1}}^{G_{j-1}})\sum_{\substack{n\geq x_1\geq x_2\\ \mathrm{\Omega}\supseteq S_{ x_1}^{G_1}\supseteq S_{x_2}^{G_2}}}\prod_{u_1\in S_{x_1}^{G_1}}{P_{u_1}^{G_0}P_{u_1}^{G_1}}\prod_{v_1\in\complement_\mathrm{\Omega}S_{x_1}^{G_1}}{(1-P_{v_1}^{G_0}P_{v_1}^{G_{1}})}
\prod_{u_2\in S_{x_2}^{G_2}} P_{u_2}^{G_2}\prod_{v_2\in\complement_{S_{x_{1}}^{G_{1}}}S_{x_2}^{G_2}}\left(1-P_{v_2} ^{G_2}\right)
\end{split}
\end{equation}

Compared Eq. (\ref{proof_2_5}) with Eq. (\ref{proof_2_2}), we can obtain that by the similar operations of Eq. (\ref{proof_2_2})-(\ref{proof_2_4}) for $|G|-1$ times, it holds:
\begin{equation}
\label{proof_2_6}
\begin{split}
P_C&= \sum_{\substack{{n\geq x_{|G|}\geq n-f}\\ \mathrm{\Omega}\supseteq S_{x_{|G|}}^{G_{|G|}}}} {\prod_{u\in S_{x_{|G|}}^{G_{|G|}}} (\prod_{j=0}^{|G|}P_u^{G_j})\prod_{v\in\complement_\mathrm{\Omega}S_{x_{|G|}}^{G_{|G|}}}(1-\prod_{j=0}^{|G|}P_v^{G_j})}
\end{split}
\end{equation}
Considering $\prod_{j=0}^{|G|}P_u^{G_j}$ as the joint success rate $P_u^J$, Eq. (\ref{as1}) is obtained.

Note that when there is no graph C in the communication structure of the consensus protocol $G$, Eq. (\ref{as1}) is not approximated since all the derivations are identity transformations. However, if the consensus protocol has graph C in the $j$-th phase, according to Eq. (\ref{rv_1A_eq_5}), $P_i^{G_j}$ is determined by the set $S_{x_{j-r_j}}^{G_{j-r_j}}$, which will vary in different summation paths. This causes that the operations of Eq. (\ref{proof_2_3})-(\ref{proof_2_4}) cannot be held. Thus when the graph of $j$-th phase is C, we use the approximation form of Eq. (\ref{rv_1A_eq_5}):
\begin{equation}
\begin{split}
\label{proofas3}
\bar{P}_i^{G_j}= \sqrt[f+1]{\frac{\sum_{w=n-f-1}^{n-1}{(\bar{P}_i^{G_j}(w))^{f+1}}}{f+1}} 
\end{split}
\end{equation}
where $\bar{P}_i^{G_j}(w)$ is
\begin{equation}
\begin{split}
\label{proofas4}
\bar{P}_i^{G_j}(w)=\frac{\sum_{\Phi_{w}\in (\Omega\setminus \{i\})}{P_i^{G_j}(\{ \Phi_w, i\})} }{\binom{n-1}{w}}
\end{split}
\end{equation}
Here $\Phi_{w}$ is a running variable with $|\Phi_w|=w$, $P_i^{G_j}(\{ \Phi_w, i\})$ is the left side in Eq. (\ref{rv_1A_eq_5}) replacing $S_{x_{j-r_j}}^{G_{j-r_j}}$ as $\{ \Phi_w, i\}$. The approximation principle is to make $\bar{P}_i^{G_j}$ independent from $S_{x_{j-r_j}}^{G_{j-r_j}}$ by taking the average of $P_i^{G_j}(S_{x_{j-r_j}}^{G_{j-r_j}})$. Specifically, Eq. (\ref{proofas4}) is the average of $P_i^{G_j}(S_{x_{j-r_j}}^{G_{j-r_j}})$ for $|S_{x_{j-r_j}}^{G_{j-r_j}}|=w+1$ and Eq. (\ref{proofas3}) is the average of order $f+1$ of $\bar{P}_i^{G_j}(w)$. Due to the averaging, when there is graph C in the consensus protocol, Eq. (\ref{as1}) is approximate and the approximation will be tight if the link reliability is close to 1.

Thus Lemma \ref{rv_th_jfv} is proved.
\end{strip}
\clearpage

\section{Proof of Theorem \ref{the:RA}}
\label{proof of the:RA}

We consider the case of BFT where $f=\lfloor n/3\rfloor$ as an example. For the cases of $n=3f+1$ and $n=3f+2$ in BFT and the cases of $n=2f$ and $n=2f+1$ in CFT, the approximations are similar. 
According to Corollary \ref{degen_simp}, Since $p_{JF,R}$ is close to 0, we have
\begin{equation}
	\begin{split}
		\label{proof_4_1}
		p_{F}&\approx \binom{n}{f+1}p_{JF,R}^{f+1}\approx \binom{n}{f+1}p_{JF,R}^{f+1}(1-p_{JF,R})^{n-f-1}
	\end{split}
\end{equation}

For $n=3f$, by applying the Stirling formula, $n!\approx\sqrt{2\pi n}(\frac{n}{e})^{n}$, we have
\begin{equation}
	\begin{split}
		\label{proof_4_2}
		&{\rm log}p_F=(f+1){\rm log}p_{JF,R} +(2f-1){\rm log}(1-p_{JF,R})
  \\ & +{\rm log}(\frac{(3f)!}{(f+1)!(2f-1)!})\\ &\approx (f+1){\rm log}p_{JF,R} +(2f-1){\rm log}(1-p_{JF,R}) \\ & +{\rm log}(\sqrt{\frac{3}{f}}3^{3f}(\frac{1}{2})^{2f})
	\end{split}
\end{equation}

In Eq. (\ref{proof_4_2}), the linear term and non-linear term with respect to $f$ can be obtained as:
\begin{equation}
	\begin{split}
		\label{proof_4_3}
		{\rm log}p_F=&({\rm log}p_{JF,R}+2{\rm log}p_{J,R}+3{\rm log}3-2{\rm log}2)\cdot f\\ &+{\rm log}(\frac{\sqrt{3}p_{JF,R}}{\sqrt{\pi}(1-p_{JF,R})})-\frac{1}{2}{\rm log}(f)
	\end{split}
\end{equation}

It could be seen that $({\rm log}p_{JF,R}+2{\rm log}p_{J,R}+3{\rm log}3-2{\rm log}2)\cdot f+{\rm log}(\frac{\sqrt{3}p_{JF,R}}{\sqrt{\pi}(1-p_{JF,R})})$ is the linear part in Eq. (\ref{proof_4_2}) while $-\frac{1}{2}{\rm log}f$ is the non-linear part. The derivative of ${\rm log}p_{F}$ with respect to $f$ is ${\rm log}p_{JF,R}+2{\rm log}p_{J,R}+3{\rm log}3-2{\rm log}2-\frac{1}{2f}$. Since $p_{JF,R}$ is close to 0, the impact of $-\frac{1}{2f}$ on the linearity is minor. Thus the non-linear part $-\frac{1}{2}{\rm log}f$ could be neglected or only be considered as the complement term to decrease the approximation error.
For the cases of $n=3f+1$ and $n=3f+2$ in BFT and the cases of $n=2f$ and $n=2f+1$ in CFT, they could be similarly proved through the Stirling formula to obtain the linearity term. Thus Theorem \ref{the:RA} has been proved.

\end{document}